%% file: Weighted-CG.tex
\documentclass[11pt]{article}
\usepackage[a4paper,margin=1in]{geometry}

\title{Fair Interventions in Weighted Congestion Games}

\author{Miriam Fischer\thanks{Department of Computing, Imperial College London, UK; {\tt m.fischer21@imperial.ac.uk}.} \and Martin Gairing\thanks{Department of Computer Science, University of Liverpool, UK; {\tt gairing@liverpool.ac.uk}.} \and Dario Paccagnan\thanks{Department of Computing, Imperial College London, UK; {\tt d.paccagnan@imperial.ac.uk}.}}

\date{}

\input{parts-frontmatter/packages.tex}

\begin{document} 
\maketitle
\thispagestyle{empty}
\input{parts-body/abstract.tex}

\maketitle

\newpage
\pagenumbering{arabic} 
\input{parts-body/introduction.tex}

\input{parts-body/fairness.tex}

\input{parts-body/opttaxes.tex}

\input{parts-body/hardnessapprox.tex}

\newpage

\input{Weighted-CG.bbl}
\newpage
\input{parts-appendix/appendix.tex}

\end{document}

%% file: parts-frontmatter/packages.tex
\usepackage[utf8]{inputenc}
\usepackage{amsmath}
\usepackage{amsthm}
\usepackage{amsfonts}
\usepackage{bbm} 
\usepackage{xcolor}
\usepackage{pgfplots}
\usepackage{graphicx}
\usepackage{xfrac}
\usepackage{tikz}
\usepackage{hyperref}
\usepackage[capitalise]{cleveref}
\usepackage{xfrac}
\usepackage[font=small, tablename=Table]{caption}
\usepackage{subcaption}
\usepackage{tcolorbox}
\usepackage{tikz-network}
\usepackage{textpos}
\usepackage{xfp}
\usepackage{pgffor}
\usepackage{siunitx}
\usepackage{multirow}
\usepackage{wrapfig}
\usepackage{enumitem}
\usepackage{comment}
\usepackage{soul}
\usepackage{bbm}
\usepackage{xspace}
\usepackage{cleveref}
\usepackage{makecell}
\usepackage{booktabs}
\usepackage{threeparttable}
\usepackage{physics}

\definecolor{colourabc}{cmyk}{.9,.5,0,.35}
\newtheoremstyle{mystyle}%
{5pt}%
{3pt}%
{}%
{}%
{\bfseries \color{colourabc}}%
{.}%
{.5em}%
{}%
\theoremstyle{mystyle}

\makeatletter
\renewcommand{\paragraph}{%
  \@startsection{paragraph}{4}%
  {\z@}{1.25ex \@plus 0.5ex \@minus .2ex}{-0.5em}%
  {\normalfont\normalsize\bfseries}%
}
\makeatother

\makeatletter
\renewcommand{\paragraph}{%
  \@startsection{paragraph}{4}%
  {\z@}{1.25ex \@plus 0.5ex \@minus .2ex}{-0.5em}%
  {\sffamily \mdseries \color{colourabc}}%
}
\makeatother

\newtheorem{claim}{Claim}

\newtheorem{lemma}{Lemma}

\newtheorem{theorem}{Theorem}
\newtheorem{corollary}{Corollary}

\newtheorem{sassumption}{Assumption} %

\newcommand{\poa}{\mathrm{PoA}}

\newcommand{\be}{\begin{equation}}
\newcommand{\ee}{\end{equation}}
\newcommand{\systemcost}{social cost\xspace}%
\newcommand{\socialcost}{social cost\xspace}%

\newcommand{\mc}[1]{\mathcal{#1}}
\newcommand{\mb}[1]{\mathbb{#1}}

\newcommand{\PNEcost}{{\tt{NECost}}}

\newcommand{\mincost}{{\tt{MinCost}}}
\newcommand{\SC}{{\rm{SC}}}
\newcommand{\C}{{\rm{C}}}
\newcommand{\opt}[1]{{#1}^{\rm opt}}
\newcommand{\NE}[1]{{#1}^{\rm ne}}

\newcommand{\NP}{{\sf NP}}
\newcommand{\Pclass}{{\sf P}}
\pgfplotsset{compat=1.15}

\newlength{\mynegspace}
\newcommand{\mygezero}{>0}

\newcommand{\mincon}{\texttt{MinSC}}

\renewcommand{\star}{\ast}

\newcommand{\mya}[2]{\alpha_{#1}^{#2}}
\newcommand{\myb}{\beta}
\newcommand{\ca}{b}
\newcommand{\tildemya}[2]{A_{#1}^{#2}}
\renewcommand{\d}{d} %
\renewcommand{\k}{k} %
\newcommand{\D}{D}   %
\renewcommand{\j}{j} %
\newcommand{\n}{n}   %
\newcommand{\p}{p}   %
\renewcommand{\i}{i}   %
\newcommand{\mydelta}{\Delta} 

\newcommand{\M}{\mathcal{M}}

\usetikzlibrary{shapes.misc}
\tikzset{cross/.style={cross out, draw, 
         minimum size=2*(#1-\pgflinewidth), 
         inner sep=0pt, outer sep=0pt}}
         
\usepackage{xr}
\definecolor{pnasblueback}{RGB}{205,217,235}

%% file: parts-body/abstract.tex
\begin{abstract}
In this work we study the power and limitations of fair interventions in weighted congestion games. Specifically, we focus on interventions that aim at improving the equilibrium quality (price of anarchy) and are fair in a suitably defined sense.
Within this setting, we provide three key contributions.
First, we show that no fair intervention can reduce the price of anarchy below a given factor depending solely on the class of latencies considered. Interestingly, this lower bound is unconditional, i.e., it applies regardless of how much computation interventions are allowed to use.
Second, we design a taxation mechanism that is fair and achieves a price of anarchy matching this unconditional lower bound, all the while being polynomial-time computable.
Third, we show that no intervention (fair or not) can achieve a better approximation if polynomial computability is required.  We do so by proving that the minimum social cost is \NP-hard to minimize below a factor identical to the one previously introduced. In doing so, our work shows that the algorithm proposed by Makarychev and Sviridenko \cite{Makarychev18} to tackle optimization problems with a ``diseconomy of scale'' is optimal, and provide a novel way to derandomize its solution via equilibrium computation.
\end{abstract}

%% file: parts-body/introduction.tex
\section{Introduction}
The overall performance of many distributed systems often depends on the choices made by myopic users, who take decision based on their individual interests and without regard to the overall performance. In this context, the presence of selfish users often leads to suboptimal outcomes and the celebrated notion of \emph{price of anarchy} \cite{koutsoupias1999worst} %
constitutes a commonly employed metric to quantify the resulting performance degradation. 

Against this backdrop, (weighted) congestion games \cite{rosenthal1973class} have played a staring role in the development of a rich body of literature primarily focusing on quantifying the price of anarchy, thanks to which we now have a complete characterization of this metric. 
However, within this setting, much less is known regarding the ability of \emph{interventions} -- a means for the system designer to modify the users' perceived costs -- to improve the equilibrium quality. {Indeed, although a multitude of approaches have been proposed}, including coordination mechanisms \cite{ChristodoulouKN09}, Stackelberg strategies \cite{Fotakis10a}, taxation mechanisms \cite{caragiannis2010taxes,BiloV19,paccagnan2021optimal}, signaling \cite{BhaskarCKS16,HoeferKKD22}, network design techniques \cite{GairingHK17,Roughgarden06}, cost-sharing methods \cite{GairingKK20,GkatzelisKR14}, it is still unclear how powerful interventions are in reducing the price of anarchy, especially when polynomial computability is required.

Concurrently, there has also been a surge of interest in developing algorithms that are \emph{fair} in the sense that they do not discriminate between players with similar or identical attributes, with examples ranging from classification \cite{dwork2012indirect} to online decision-making \cite{gupta2021implication} and allocation of indivisible items \cite{aziz2022equilibria}. It is therefore natural to study how effective \emph{fair interventions} (see \cref{subsec:weighted_intro} for a definition) are in reducing the price of anarchy. In particular, it is interesting to study whether the fairness constraint imposes any additional limitation on the achievable price of anarchy, and how such additional requirement intertwines with the computational complexity of the resulting interventions.

\smallskip
These questions can be summarized as follows:
\begin{itemize}[leftmargin=5.5mm]
	\item {\it How effective are fair interventions in reducing the price of anarchy?}
	\item {\it \mbox{Do fairness requirements limit the ability of interventions to improve the price of anarchy?}}
\end{itemize}  

\smallskip
In this paper we settle these questions for the class of weighted congestion games. We do so both in the setting where no limitation is imposed on the intervention's computational requirement, and in the case where polynomial computability is desired. Perhaps interestingly, our work shows that when efficient computation is needed, fair interventions are as a powerful mean to reduce the price of anarchy as any other. 

\subsection{Our Contributions and Their Implications}
\label{subsec:contributions}
Three major technical contributions substantiate and expand upon the previous claims:
\begin{itemize}[leftmargin=5.5mm]
\item[C1)] {\bf Unconditional Lower Bound for Fair Interventions.} We give a lower bound on the achievable price of anarchy of fair interventions that is independent of their computational complexity, but depends solely on the class of latency functions employed (\textit{\cref{th:approximation}}). Stated differently, we show that no fair intervention can achieve a price of anarchy below this bound. %
To the best of our knowledge, our work is the first to analyze unconditional lower bounds for fair interventions.

\item[C2)] \textbf{Optimal and Efficiently-computable Fair Interventions%
.} 
We provide a framework to design fair interventions, 
in the form of taxation mechanisms, whose price of anarchy matches the lower bound previously introduced, and show that such interventions are polynomially-computable (\textit{\cref{thm:opttaxes}}).
Further, for the widely studied class of polynomial latency functions, we provide explicit expressions of the optimal mechanism (\textit{\cref{thm:recursionpolynomials}}). Even without fairness constraints, the problem of designing polynomially-computable interventions to reduce the price of anarchy has not been previously studied beyond linear latency functions~\cite{caragiannis2010taxes}.

\item[C3)] \textbf{Limitations of Unfair Interventions.} We complement the previous results by showing that no intervention, regardless of whether it is fair or not, can reduce the price of anarchy below the factor previously derived, if polynomial computability is required. We prove this result by showing that the problem of minimizing the social cost is \NP-hard to approximate below a factor identical to that already obtained (\textit{\cref{thm:hardness}}). No tight hardness of approximation was known for weighted congestion games.%
\end{itemize}
\smallskip

\noindent Taken together these contributions show that: 
i) The proposed taxation mechanisms are optimal amongst all fair intervention in the sense that no fair intervention can induce a better price of anarchy, regardless of its computational complexity;
ii) The proposed taxation mechanisms are also optimal within all polynomially-computable interventions, fair or not, and more generally also against any polynomially-computable algorithm. 
This is remarkable as, prior to this work, it was unclear whether (fair) interventions were a sufficiently powerful tool to achieve optimal approximation in weighted congestion games. Our paper answers this question in the positive.

Finally, although the following observations are direct consequences of the previous contributions, we highlight them separately, as we believe they may be of independent interest.

\begin{itemize}[leftmargin=5.5mm]
\item[C4)] {\bf Optimal Approximation, and Matching Hardness.} Makarychev and Sviridenko \cite{Makarychev18} study approximation algorithms for optimization problems with a ``diseconomy of scale'', to which weighted congestion games belong. Interestingly, the hardness factor we obtain in C3) matches the approximation factor derived by \cite{Makarychev18} and thus fully settles the approximability of this important class of problems. Finally, we note that by combining our results in C2) with existing algorithms for the computation of correlated equilibria (e.g., no-regret dynamics), one obtains a polynomial-time algorithm to minimize the social cost that achieves the best possible approximation (\textit{\cref{cor:polyalgo}}). This approach provides an alternative to the randomized rounding techniques proposed by Makarychev and Sviridenko \cite{Makarychev18}.
\end{itemize}
\smallskip

\smallskip
\noindent 
Three final comments are in order.
First, note that our unconditional lower bound only requires a very mild notion of fairness (see \cref{subseclabel:fairness}, Interventions \& Fairness for details on its definition). This makes the unconditional lower bound very strong, as the impossibility result we derive holds over a large class of interventions, and strengthening the notion of fairness can only worsen the lower bound. At the same time, we note that the taxation mechanism we present satisfies a much stronger fairness notions, as players using the same resource are treated equally even if they have different weights. 
Second, we observe that the price of anarchy obtained by our taxation mechanism holds both for mixed Nash equilibria and correlated/coarse correlated equilibria. This extension is significant as it gives performance bounds that apply also when players revise their action and achieve low regret \cite{roughgarden2015intrinsic}.
Third, our work also settles the question of whether approximating the minimum social cost is strictly harder in weighted congestion games compared to their unweighted variant%
. Perhaps interestingly, our work shows that the two problems are equally difficult to approximate only when the class of latency functions considered is closed under \textit{abscissa scaling} (\cref{lem:factor-weighted-unweighted}), which includes the important class of polynomial latency functions.\footnote{A class of latency functions $\mc{L}$ is closed under abscissa scaling if, for any function $\ell \in \mc{L}$ and $\alpha \geq 0$, the function $h$ with $h(x)=\ell(\alpha x)$ is in $\mc{L}$.} 
\subsection{Techniques} 
Three separate sets of techniques are needed to prove our results.

The lower bound for fair interventions is shown using, at its core, a symmetry argument. Specifically, the lower bound is obtained by exhibiting an instance of weighted congestion games where any fair intervention induces high cost at the worst mixed Nash equilibrium, relative to the optimal allocation. Interestingly, the instance utilized consists of a network with parallel links and identical players. In these settings, fair interventions maintain the symmetry of the game owing to their definition, so that a symmetric mixed Nash equilibrium must exist \cite{nashgames}. One then concludes by showing that any symmetric mixed profile, whether or not an equilibrium, has the desired high social cost.

Taxation mechanisms are a natural type of intervention often considered to reduce the price of anarchy, and their \emph{fairness} follows readily from their definition. Contrary to that, providing a recipe to design taxation mechanisms that are both efficiently computable \emph{and} achieve the desired approximation is significantly more challenging. We achieve this goal following a two-pronged approach. 
First, we compute a fractional solution to a Linear Programming (LP) relaxation for the problem of minimizing the social cost, and then design taxes which, using the solution as input, make sure that the resulting equilibria achieve the desired approximation. 

In more detail, we leverage the LP relaxation of \cite{Makarychev18} to minimizing social cost in combinatorial optimization problems with a diseconomy of scale. While this relaxation has an exponential number of decision variables (for each resource there is a variable for each subset of players), the resulting LP can be solved in polynomial time to arbitrary precision using the ellipsoid method. We adapt this relaxation to weighted congestion games, and show how its (non-integral) solution can be used as input to our taxation mechanism, which defines taxes for each resource.
Using critical inequalities based on stochastic orders of sums of Poisson random variables, we prove that this taxation mechanism enforces equilibria with the desired price of anarchy. We note that this is challenging, as the LP relaxation of \cite{Makarychev18} was designed to compute an allocation with low social cost in a \emph{centralized} setting, and not in the game-theoretic setting we consider here. Our result thus shows that, whether agents are strategic or not, it is possible to achieve the same optimal approximation.

For polynomial latency functions, we present closed-form expressions of the taxation mechanism in \cref{sec:poly}. We do so by showing that the moments of a weighted sum of independent Poisson random variables can be computed through a binomial-based recursion, which might be of independent interest. The proof heavily relies on moment generating functions and generalizes a result by Gould and Quaintance \cite[Eq 4.2]{gould2007linear}.

The result regarding hardness of approximation is obtained through an approximation preserving reduction \cite{AusielloPashos2014} from the corresponding problem for unweighted congestion games \cite{paccagnan2021optimal}. Although weighted congestion games are significantly more expressive than their unweighted variant, we are able to apply such approximation preserving reduction owing to following crucial observation: weighted congestion games where players have \emph{identical} weights are sufficient to achieve the largest possible inapproximability factor, once again highlighting the \mbox{importance of symmetry}. 

\medskip
Finally, we note that this paper differs significantly from recent work on the design of taxes for \emph{unweighted} congestion games \cite{paccagnan2021optimal}. 
First, the authors of \cite{paccagnan2021optimal} do not consider fairness constraints and therefore do not study how such requirement impacts the achievable performance, e.g., the unconditional lower bound in \cref{th:approximation}.
Second, their work is limited to the unweighted setting, and, perhaps unsurprisingly, their techniques do not carry over to the setting considered here. For example, in the unweighted case, they obtain taxation mechanisms through the solution of a minimization problem where the cost on each resource is replaced by a Poisson approximant. While this problem can be solved in polynomial time for the unweighted case, it can not be used here as it gives rise to a non-convex optimization problem for the weighted case.
Third, obtaining explicit expression for optimal taxation mechanisms requires new ideas, even for the widely studied class of polynomial latencies. For example, in \cref{lem:degree0equal}, we show how generalized Bell numbers can be constructed by summing a sequence of numbers (different from the commonly-employed Stirling numbers of the second kind) naturally arising in combinatorics.
Lastly, as we will show, the \NP-hardness factor for minimizing the social cost is also different in the weighted and unweighted settings. 

\subsection{Related Work}
\paragraph*{Price of Anarchy}
Since its introduction \cite{koutsoupias1999worst}, the \emph{price of anarchy} has been extensively studied in various setting, and for both unweighted and weighted congestion games we now have a rich theory. This includes exact bounds for congestion games with linear \cite{awerbuch2005price,christodoulou2005price}  and 
polynomial \cite{aland2006exact} latency functions.  These results led to the development of Roughgarden's \emph{smoothness framework} \cite{roughgarden2015intrinsic}, which elegantly distilled and formulated previous ideas for bounding the price of anarchy in a unified framework. Roughgarden \cite{roughgarden2015intrinsic} also showed that upper bounds achieved through the smoothness framework extend all the way to \emph{coarse correlated equilibria}, which can be computed by \emph{no-regret} algorithms. Bhawalkar et al.~\cite{bhawalkar2014smoothness} extended those results to weighted congestion games. 

\paragraph*{Fairness} 
While fairness notions have been widely studied, one of the most commonly employed is that of \textit{individual fairness} \cite{dwork2012indirect}, which originated in the design of classification algorithms. Such notion requires that individuals with similar attributes receive similar outcomes. %
In this respect, our notion of fair interventions (see \cref{subsec:weighted_intro} for a definition) 
is closely aligned to the former, with the crucial difference that we do not define similarity as a continuous metric, but rather only with regards to identical players. In this respect, our notion of fairness imposes weaker constraints than that of individual fairness, resulting in a very strong unconditional lower bound. Indeed, many of the interventions studied in the literature implicitly satisfy our fairness definition and, therefore, must obey our unconditional lower bound. These include cost-sharing mechanisms \cite{GairingKK20,GkatzelisKR14}, public signaling in singleton congestion games \cite{zhou2022algorithmic}, cost-balancing tolls in atomic symmetric network congestion games \cite{fotakis2008cost} and taxation mechanisms by \cite{caragiannis2010taxes,BiloV19,paccagnan2021atomic}. At the same time, we observe that the interventions we propose satisfy a much stronger fairness notion as players selecting the same resource are treated equally, even if they have different weights.

Finally, we observe that, within the realm of congestion games, the fairness of \emph{allocations} has been studied for both the non-atomic \cite{correa2007networks,roughgarden2002unfair} and the atomic  \cite{ChakrabartyMN05,kleinberg1999fairness} variants of congestion games, though with significant differences. Most notably, these works focus on the fairness of a given allocation, and not on how that allocation can be induced by fair interventions.

\paragraph*{Interventions \& Taxation Mechanisms}
Different approaches, such as coordination mechanisms \cite{ChristodoulouKN09}, Stackelberg strategies \cite{Fotakis10a}, taxation mechanisms \cite{caragiannis2010taxes}, signaling \cite{BhaskarCKS16}, cost-sharing strategies \cite{GairingKK20,GkatzelisKR14}, have been proposed to cope with the performance degradation associated to selfish decision making.
Amongst them, taxation mechanisms have attracted significant attention thanks to their ability to indirectly influence the resulting system performance. Their design was initiated by Caragiannis et al.~\cite{caragiannis2010taxes} for linear (weighted) congestion games, and further studied by \cite{BiloV19, paccagnan2019incentivizing} in the polynomial setting. Limitedly to linear congestion games, Caragiannis et al.~\cite{caragiannis2010taxes} provide efficiently-computable taxes whose price of anarchy matches the one we achieve. 
For polynomial congestion games, Bil\`o and Vinci \cite{BiloV19} also propose taxation mechanisms with a price of anarchy matching ours. However, there is a fundamental difference: the taxes \cite{BiloV19} proposes require knowledge of an \emph{optimal} allocation, and therefore can not be efficiently computed, as minimizing the social cost is \NP-hard. In contrast, our taxation mechanism achieves the same price of anarchy whilst also being \emph{efficiently} computable.

\paragraph*{Computational Lower Bounds}
Meyers and Schulz \cite{MeyersS12} show that minimizing the social cost in unweighted congestion games is \NP-hard to approximate for general convex latency functions, while Roughgarden~\cite{Rough_barrier} showed that computational lower bounds on minimizing the social cost translate to lower bounds on the price of anarchy. Specifically, for polynomial latency functions of maximum degree $D$ and non-negative coefficients, he provided the first lower bound, parametrized by $D$, on approximating \systemcost.
For unweighted congestion games, tight hardness factors for approximating the minimum social cost have recently been given by Paccagnan and Gairing \cite{paccagnan2021optimal} for general latency functions. No tight hardness factor was known for weighted congestion games. 

\paragraph*{Unconditional Lower Bounds}
The aforementioned lower bounds are conditional on polynomial-time computability and do not consider fairness constraints. Unconditional lower bounds have been obtained for Stackelberg strategies in unweighted congestion games with linear latencies \cite{Fotakis10a, bilo2019stackelberg} and, in weighted congestion games, for cost-sharing mechanisms \cite{kollias2015equilibria, GairingKK20} and taxes for linear latency functions \cite{caragiannis2010taxes}. However, these works restrict to specific classes of interventions, and \textit{do not} consider fairness. Our work significantly departs from these as it provides - for general latency functions - unconditional lower bounds for \textit{any} intervention that is \textit{fair}. For linear latency functions, we recover the unconditional lower bound on taxes by Caragiannis et al.~\cite{caragiannis2010taxes}, and thus provide a fairness interpretation of their work.

\subsection{Weighted Congestion Games and Taxation Mechanisms}
\label{subsec:weighted_intro}
\paragraph*{Notation}
Let $\mb{N},\mb{N}_0,\mb{R},\mb{R}_{\ge 0},\mb{R}_{\mygezero}$ denote the sets of natural numbers, natural numbers including zero, real numbers, non-negative real numbers, and positive real numbers. Given $m\in\mb{N}$, let $[m]$ denote the set $\{1,\dots,m\}$. Throughout we will denote with $\text{Poi}(x)$ a Poisson distribution with parameter $x$ and with $\text{Bin}(n,p)$ a Binomial distribution with $n$ trials and success probability $p$ for each trial. We will also use $P\le_{cx}Q$ to denote that distribution $P$ is smaller-equal than distribution $Q$ in the convex order sense, that is, $\mb{E}_P[f(P)] \le \mb{E}_ Q[f(Q)]$ for all convex functions $f:\mb{R}\rightarrow\mb{R}$ \cite{shaked2007stochastic}.

\paragraph*{Weighted Congestion Games} 
In a \emph{weighted congestion game} we are given a set of players $[N]=\{1,\dots,N\}$, and a set of resources $\mc{R}$. Each player $i\in[N]$ has a \emph{weight} $w_i\in\mathbb{R}_{>0}$ and a \emph{strategy set}  
$\mc{A}_i\subseteq 2^{\mc{R}}$ from which player $i$ choses an \emph{action} $a_i\in \mc{A}_i$. In an \emph{allocation} $a=(a_1, \dots, a_N)\in\mc{A}$, each player chooses an action from their strategy profile, where we denote by $ \mc{A}= \mc{A}_1\times\ldots \times \mc{A}_N$ the set of feasible allocations.
 
Given an allocation $a$, denote $x_r(a)$ the \emph{load} of players currently selecting resource $r\in\mc{R}$, i.e., $x_r(a)=\sum_{i\in[N]:r\in a_i} w_i$.
The cost for using each resource $r\in\mc{R}$ depends only on the load of players concurrently selecting that resource, and is defined by a \emph{latency function} $\ell_r :\mb{R}_{>0} \rightarrow \mb{R}_{\ge 0}$.
Throughout the paper we assume that we can evaluate $\ell_r(x)$ at any given load $x$ in polynomial time with respect to the length of the input representation of $\ell_r$ and $x$.
The cost of a player $i$ is obtained by summing the costs of all resources selected, that is $\C_i(a)=\sum_{r\in a_i} \ell_r(x_r(a))$. Throughout this paper, we define the shorthand $c_r(x)= x \ell_r(x)$ for all $x$ and $r\in\mc{R}$.

In a \emph{mixed Nash equilibrium} every player $i$ chooses a probability distribution $\sigma_i$ over their strategy set $\mc{A}_i$, such that no player can unilaterally improve by changing to another action; formally, for the product distribution $\sigma=\sigma_1 \times\dots\times \sigma_N$, 
\(
\mb{E}_{a\sim\sigma}\left[{C}_i(a)\right]
\le 
\mb{E}_{a_{-i}\sim\sigma_{-i}}\left[{C}_i(a_i',a_{-i})\right], \)
for every $i\in[N]$ and $a_i'\in\mc{A}_i$.
Similarly a joint probability distribution $\sigma$ over allocations is a \emph{coarse correlated equilibrium} if for every $i\in[N]$ and $a_i'\in\mc{A}_i$, 
\(
\mb{E}_{a\sim\sigma}\left[{C}_i(a)\right]
\le 
\mb{E}_{a\sim\sigma}\left[{C}_i(a_i',a_{-i})\right].
\)
Finally, the \emph{\systemcost} represents the weighted sum of the players' costs
\begin{equation}
\SC(a)=
\sum_{i \in [N]} w_i C_i(a) = \sum_{r\in\mc{R}} x_r(a) \ell_r(x_r(a)).
\label{eq:systemcost}
\end{equation}

For a mixed or coarse correlated profile $\sigma$ the social cost is defined accordingly as $\mb{E}_{a\sim\sigma}\left[{SC}(a)\right]$.

Throughout the manuscript, we make the following standard assumption on each latency function of the underlying congestion game.
\begin{sassumption} \label{assumptionlatency}
The function $\ell:\mb{R}_{\ge0}\rightarrow\mb{R}_{\ge 0}$ is non-decreasing and semi-convex.\footnote{The function $\ell(x)$ is semi-convex in its domain if the function $c(x)=x\ell(x)$ is convex, in the discrete sense, in the same domain.}
\end{sassumption}

Further, we denote with $\mc{G}$ the set of all congestion games where all latency functions $\{\ell_r\}_{r\in\mc{R}}$ belong to a given set of functions $\mc{L}$. Given an instance $G=(N,\{w_i\}_{i=1}^N,\mc{R},\{\mc{A}_i\}_{i=1}^N,\{\ell_r\}_{r\in\mc{R}})$ of congestion game, we denote with $\mincon$ the problem of minimizing the \socialcost in \eqref{eq:systemcost}.

\paragraph*{Interventions \& Fairness} \label{subseclabel:fairness}
As self-interested decision making often deteriorates the system performance, interventions aim to ameliorate this issue. Formally, an intervention $I:G \mapsto \bar G$ is a map that associates a congestion game $G$ to a new game $\bar G$ (not necessarily a congestion game), where the players are the same, but their perceived cost has been redefined to be $\bar{C}_i(a)$, and their action sets have been (potentially) reduced to $\bar{\mc{A}}_i\subseteq\mc{A}_i$. We emphasize that although $I$ changes players' \textit{perceived} cost from $C_i(a)$ to $\bar{C}_i(a)$, the \systemcost \eqref{eq:systemcost} remains unchanged. 
A \textit{fair} intervention gives identical players identical treatment. That is, if $i$ and $j$ have identical action sets and weights, then $\bar{\mc{A}}_i=\bar{\mc{A}}_j$ and $\bar C_i(a)=\bar C_j(\pi_{ij}(a))$, with $\pi_{ij}(a)$ permuting $a_i$ and $a_j$ in allocation $a$. This implies that if a fair intervention is applied to a symmetric congestion game $G$, the resulting game $\bar G$ is also symmetric.\footnote{In symmetric games, all players have identical action sets, and the cost a player incurs only depends on its action, and the actions played by all players, but not on who plays them, see \cite{vorobeychik2004,2010symmetric}.} Note that this is a very mild request, and corresponds to considering interventions that do not distinguish the players identity in $\bar G$ if these where not distinguishable in $G$. When defined this way, fair interventions include a wide range of commonly employed mechanisms to improve the equilibrium performance, most notably taxation mechanisms, cost-sharing mechanisms, public signaling, and many others.

\paragraph*{Taxation Mechanisms}
In Section \ref{sec:taxes}, we will study \textit{taxation mechanisms}, a specific type of intervention that only modifies the players' costs. Formally, a taxation mechanism $T:G\times r\rightarrow \tau_r$ associates an instance $G$, and a resource $r\in\mc{R}$ to a taxation function $\tau_r:\mb{R}\rightarrow\mb{R}_{\ge0}$. Note that each taxation function $\tau_r$ is (possibly) congestion-dependent, that is, it associates the load on resource $r$ to the corresponding tax. As a consequence, each player $i$ experiences a cost factoring both the cost associated to the chosen resources, and the \mbox{tax, i.e.,}
\[
\bar C_i(a)=\sum_{r\in a_i} [\ell_r(x_r(a))+\tau_r(x_r(a))].
\]
We note that defined this way, taxation mechanisms are fair interventions as all players using a specific resource $r$ experience the same tax $\tau_r$.

\paragraph*{Price of Anarchy}
As typical in the literature, we measure the performance of a given intervention $I$ using the ratio between the \systemcost incurred at the worst-performing outcome and the minimum \systemcost. Given the self-interested nature of the players, an outcome is commonly described by any of the following classical equilibrium notions: pure or mixed Nash equilibria, correlated or coarse correlated equilibria. %
However, we recall that within the class of weighted congestion games, existence is guaranteed for mixed Nash equilibria and correlated or coarse correlated equilibria, but not necessarily for pure Nash equilibria. When considering mixed Nash equilibria, the performance of an intervention $I$ is gauged using the notion of \emph{price of anarchy} \cite{koutsoupias1999worst}, i.e., 
\begin{equation}
\poa(I) = \sup_{G\in\mc{G}} \frac{\PNEcost(G,I)}{\mincost(G)},
\label{eq:poadef}
\end{equation}
where $\mincost(G)=\min_{a\in\mc{A}} \SC(a)$ is the minimum \socialcost for instance $G$, and $\PNEcost(G,I)$ denotes the highest expected \socialcost at a mixed Nash equilibrium obtained when employing the intervention $I$ on the game $G$. By definition, $\poa(I)\ge1$ and the lower the price of anarchy, the better performance $I$ guarantees. While it is possible to define the notion of price of anarchy for each and every equilibrium class mentioned, we do not pursue this direction, as we will show that all these metrics coincide within our setting. Thus, we will simply use $\poa(I)$ to refer to the efficiency values of \emph{any} and \emph{all} these equilibrium classes. Finally, we observe that, while interventions influence the players' perceived cost, they do not impact the expression of the \systemcost, which is still of the form in \eqref{eq:systemcost}.

\subsection{Roadmap}
The remainder of the manuscript is organized as follows. In Section~\ref{sec:fairlowerbound}, we provide the unconditional lower bound for fair interventions. In Section~\ref{sec:taxes}, we provide fair, polynomially computable taxes that achieve a price of anarchy matching the lower bound previously derived. We specialize this result to the case of polynomial latency functions in Section~\ref{sec:poly}. In Section~\ref{sec:hardness} we show hardness of approximating the minimum social cost and limitations of unfair interventions.

%% file: parts-body/fairness.tex
\section{Lower bound on Price of Anarchy of Fair Interventions} \label{sec:fairlowerbound}
Our first result shows that fair interventions incur a fundamental limitation in their ability to improve the equilibrium performance, and that such limitation is independent of how much computational power the intervention is given, i.e., $I(G)$ might even require exponential computation in the size of $G$. This is the case already in the setting where all resources have the same latency $\ell$ and all players are symmetric, i.e., they have the same action sets and weights.

\begin{theorem} \label{th:approximation}
In congestion games where all resources have the same latency function $\ell$, there exists no fair intervention that induces a price of anarchy lower than $\rho_\ell - \varepsilon$, for any $\varepsilon>0$, where
\be
\rho_\ell=\sup_{t>0}\frac{\mb{E}_{P\sim\text{Poi}(1)} [tP\ell(tP)]}{t\ell(t)},
\label{eq:approxfact}
\ee
and we define $\ell(0)=0$. {If $\rho_\ell=\infty$, the price of anarchy of any fair intervention is unbounded.
 }

\end{theorem}

\begin{proof}
We prove the claim by exhibiting an instance of congestion game $G$ with $m$ resources and $m$ players for which no fair intervention has a price of anarchy smaller that $\rho_\ell-\varepsilon$, where $\varepsilon$ can be made arbitrarily small by letting $m$ grow unbounded. 

Consider a singleton congestion game $G$ with $m$ identical resources $r_1,\dots,r_m$, all having the same latency function $\ell$. The game $G$ has $m$ identical players, in the sense that each player has identical action set $\mc{A}_i = \{\{r_1\},\{r_2\},\dots,\{r_m\}\}$ and identical weight $w_i = w$, for all $i$, %
where 
\begin{align*}
w &\in  \arg \sup_{t>0}\frac{\mb{E}_{P\sim\text{Poi}(1)} [tP\ell(tP)]}{t\ell(t)}.
\end{align*}

Recall the shorthand $c(x)=x\ell(x)$. Observe that the allocation $a$ where each player selects a different resource has social cost $\SC(a) = \sum_{r \in \mc{R}} c(w) = m c(w)$. Further, observe that, since the game $G$ is symmetric\footnote{It is immediate to observe that every congestion game where players have the same action set and the same weight is a symmetric game according to the definition in, e.g., \cite{vorobeychik2004,2010symmetric}.} and $I$ is a fair intervention, $\bar G$ is also symmetric. As such, $\bar G$ admits a symmetric mixed Nash equilibrium \cite[Theorem 2]{nashgames}, i.e., a mixed Nash equilibrium where each player plays the same probability distribution over the actions. 

To obtain the desired lower bound on the price of anarchy, we will show that \emph{any} symmetric mixed assignment (regardless of whether it is a Nash equilibrium or not) has high social cost compared to $a$. Specifically, we show that the social cost of any symmetric mixed assignment is at least $(\rho_\ell - \varepsilon)\SC(a)$, implying that the cost of any symmetric mixed Nash equilibrium must also be at least $(\rho_\ell - \varepsilon)\SC(a)$. Towards this goal, denote with $y\in\mb{R}^{m}$ a probability vector, where $y_r$ denotes the probability that any player chooses resource $r\in\mc{R}$.\footnote{Since players have singleton actions, every action corresponds to a resource, and thus we define a symmetric mixed profile directly over the set of resources.} It follows that the expected social cost of any such mixed profile is 
\[
\mb{E}_{X_r\sim \text{Bin}(m,y_r)}\left [\sum_{r\in\mathcal{R}} c(wX_r)\right],
\]
which holds since $X_r=\sum_{i=1}^m X_{r,i}$ where $X_{r,i}\sim \text{Ber}(y_r)$ and $X_{r,i},X_{r,j}$ are independent for $i\neq j$, as each player chooses independently resource $r$ with probability $y_r$. 

To conclude the proof, we will now show that the latter cost is no-lower than that achieved by the uniform distribution $y_r=1/m$ for all $r\in\mc{R}$. This will be sufficient to conclude, as it implies that %
\[
\frac{\mb{E}_{X_r\sim \text{Bin}(m,y_r)}\left [\sum_{r\in\mathcal{R}} c(wX_r)\right]}{mc(w)}\ge \frac{\mb{E}_{X\sim \text{Bin}(m,1/m)}\left [c(wX)\right] }{c(w)}\xrightarrow{m\rightarrow\infty}\frac{\mb{E}_{X\sim \text{Poi}(1)}\left [c(wX)\right] }{c(w)}=\rho_\ell,
\]
i.e., that the price of anarchy is no lower than $\rho_\ell-\varepsilon$, where $\varepsilon$ can be reduced to any desired positive number by letting $m\rightarrow\infty$. This holds true since the expected value over the Binomial distribution $\text{Bin}(m,1/m)$ is known to converge to the expected value over the Poisson distribution $\text{Poi}(1)$ \cite[p.56]{papoulis2002probability}.

Therefore, denoting $p(y_r)=\mb{E}_{X_r\sim \text{Bin}(m,y_r)}\left [c(wX_r)\right]$, it suffices to show that $\sum_{r\in\mc{R}} p(y_r) $ is minimized by the uniform distribution $y_r=1/m$.

To see this, note that $p_r(\cdot)$ is a convex (differentiable) function, as shown in \cite[Example 8.A.3]{shaked2007stochastic}, and that the simplex is a convex polytope. Therefore, the corresponding KKT conditions of optimality
\[
\begin{cases}
p'(y_r)+\lambda &= 0\\
\sum_r y_r&=1
\end{cases}
\] 
are readily satisfied by $y_r^*=1/m$, $\lambda^* = -p'(1/m)$, concluding the proof.
\end{proof}

%% file: parts-body/opttaxes.tex
\section{Fair, optimal, and efficiently-computable Taxes} 
\label{sec:taxes}

In this section we provide a general recipe to design efficiently-computable taxation mechanisms with a price of anarchy matching the lower bound on fair interventions previously derived. We will then specialize this approach to the case of polynomial latencies, for which we will derive explicitly given taxation mechanisms. 
We note that the taxation mechanisms we study are not only fair in the sense of \cref{subseclabel:fairness}, but satisfy a \textit{much stronger} fairness notion in that players using the same resources are treated equally, regardless of their weight.

Our approach will make use of two crucial ingredients which we now introduce: 
i) a {\it linear programming relaxation of {$\mincon$}}, and ii) a {\it parametrized taxation mechanism}. 
At its core, our approach will derive a set of parameters from an optimal solution of such linear program, and use them for the aforementioned taxation mechanism. 
Towards this goal, we make the following minimal assumption, which we require in this section and in \cref{sec:poly}.
\begin{sassumption} \label{assumptionbounded}
Players' weights $\{w_i\}_{i=1}^N$ are multiples of some factor $\delta\ge 0$ and $w_i/\delta$ are polynomially bounded in $N$ and $1/\varepsilon$, where $\varepsilon$ is the desired accuracy.
\end{sassumption}
We note that this assumption, also used in \cite{Makarychev18}, is purely technical as it if often possible to round the weights of a given game so that the resulting cost changes by a small quantity (see \cite[Theorem A.1]{Makarychev18}).
Recall the shorthand $c_r(x)= x \ell_r(x)$ for all $x$ and $r\in\mc{R}$.

\subsection{LP Relaxation of \mincon} \label{sec:LPprogram}
The first ingredient our approach employs is a suitably defined relaxation of the original (integral) {\mincon} problem. 
While it is straightforward to formulate {\mincon} as the following integer program (where $y_{i,a_i}$ equals one if player $i$ selects action $a_i\in\mc{A}_i$ and zero else),

\be
\label{eq:original_ICVXP}
\begin{aligned}
\min
\quad & \sum_{r\in\mc{R}} c_r\left(\sum_{i\in [N]} w_iv_{r,i}\right)
\\
\text{\normalfont subject to}\quad 
& v_{r,i} =\sum_{a_i\in\mc{A}_i:r\in a_{i}}y_{i,a_i} && \text{for all } r\in \mc{R}, i\in[N],\\
& \sum_{a_i\in\mc{A}_i} y_{i,a_i} = 1 && \text{for all } i \in [N],\\
& y_{i,a_i} \in \{0,1\} && \text{for all } a_i\in\mc{A}_i, i\in[N], 
\end{aligned}
\ee

we do not employ its natural relaxation obtained removing the integral constraints, but instead make use of a more refined one following an approach similar to that in \cite{Makarychev18}. 
Towards presenting the relaxation, it is convenient to introduce variables $z_{S,r}\in\{0,1\}$ for $S\subseteq [N]$, $r\in\mc{R}$, encoding what subset of players utilize resource $r$, where $\sum_{S\subseteq[N]} z_{S,r}=1$, i.e., only one subset is active. The previous program can then be equivalently written as the following linear program in the variables $z_{S,r}$, $v_{r,i}$, $y_{i,a_i}$, upon additionally restricting the variables $z_{S,r}$ to be binary, i.e., $z_{S,r}\in\{0,1\}$,
\be
\label{eq:cvxprogram}
\begin{aligned}
\min
\quad & 
\sum_{r\in\mc{R}} 
\sum_{S\subseteq [N]}
c_r\left(\sum_{i\in S}w_i \right) \cdot z_{S,r}\\
\text{\normalfont subject to}\quad  
& \sum_{S\subseteq [N]} z_{S,r} = 1 && \text{for all } r\in \mc{R},\\
& \sum_{S\subseteq [N]:i\in S} z_{S,r} = v_{r,i} &&\text{for all } r\in \mc{R}, i\in[N],\\
& z_{S,r}\ge 0 &&  \text{for all } r\in \mc{R}, S\subseteq[N],\\
& v_{r,i} =\sum_{a_i\in\mc{A}_i:r\in a_{i}}y_{i,a_i} && \text{for all } r\in \mc{R}, i\in[N],\\
& \sum_{a_i\in\mc{A}_i} y_{i,a_i} = 1 && \text{for all } i \in [N],\\
& y_{i,a_i} \ge 0 && \text{for all } a_i\in\mc{A}_i, i\in[N].
\end{aligned}
\ee

Program \eqref{eq:cvxprogram} constitutes the linear programming relaxation of  \eqref{eq:original_ICVXP} we will use in the remainder of this manuscript.
{Due to the presence of $z_{S,r}$, $S\subseteq [N]$, this linear program has exponentially many decision variables. Yet, we show in \cref{app:LPcomputation} that \eqref{eq:cvxprogram} can be solved to arbitrary precision in polynomial time, following a similar approach to that of \cite{Makarychev18}.}

\subsection{Parametrized Taxation Mechanisms}
Towards introducing the taxation mechanism, it is convenient to define the following functions
$p_r:\mb{R}^{N}\rightarrow \mb{R}$ (one for each resource) 
\be
p_r(v_r)
=
\mb{E}\left[c_r\left(\sum_{i\in [N]} w_i P_{i}\right)\right],
\label{eq:pdef}
\ee
where $P_{i}\sim \text{Poi} \left({v_{r,i}}\right)$ are independent Poisson random variables, $c_r:\mb{R}\rightarrow\mb{R}$, defined by $c_r(x)=x\ell_r(x)$, is convex, and $v_{r,i}$ denotes the $i$-th entry of the vector $v_r\in\mb{R}^N$.\footnote{In the ensuing proof, $v_{r,i}$ will represent the proportion of player $i$'s weight allocated to resource $r$.}
Throughout, we assume we can evaluate \eqref{eq:pdef} for given $v_r$ in polynomial time with respect to the input representation of the game, and show in \cref{lem:Md} that this is the case for the widely studied class weighted congestion games with polynomial latencies by obtaining an explicit expressions.

The parametrized taxation mechanism we employ  $\tau_r:\mb{R}\times \mb{R}^N\rightarrow\mb{R}$ associates each resource $r$ to a corresponding tax $\tau_r(x;v_r)$ that depends on the total load $x$ on that resource and on the previously introduced parameter $v_r$. 
Letting $\bar\ell_r(x;v_r)=\tau_r(x;v_r)+\ell_r(x)$ represent the modified latency function accounting for both the original latency and the imposed tax, we define such taxation mechanism through the following recursion

\be
x_r\bar{\ell}_r(x_r;v_r)-\sum_i {w_i}{v}_{r,i} \bar\ell_r(x_r+w_i;v_r)=c_r(x_r)- p_r(v_r)\qquad \forall x_r,v_r\ge0,
\label{eq:recursion_tax}
\ee
where we additionally set $\bar \ell_r(x_r,0)=\ell_r(x_r)$ for all $x_r\ge0$, and $\bar \ell_r(0,v_r)=0$ for all $v_r\ge0$. While implicitly defined in \eqref{eq:recursion_tax} for sake of generality, we will provide an explicit expression for $\tau_r(x;v_r)$ for the class of polynomial latencies in \Cref{sec:poly}.

\subsection{An Optimal Taxation Mechanism}
Equipped with the parametrized taxation mechanism defined in \eqref{eq:recursion_tax}, and with the linear programming relaxation of \eqref{eq:cvxprogram}, we are now ready to provide a general methodology to derive taxation mechanisms with optimal approximation. We begin with the general case, and then specialize this to polynomial latency functions. In the following, denote $P(\lambda)\sim \text{Poi}(\lambda)$ a Poisson random variable with mean $\lambda\ge0$. 

\begin{theorem}
\label{thm:opttaxes}
Consider a weighted congestion game $G$ where each latency belongs to a common set of functions $\mc{L}$  and satisfies \cref{assumptionlatency}. Let $(\bar z, \bar v, \bar y)$ be an optimal solution of  \eqref{eq:cvxprogram}.	Let $\bar \ell_r(x; \bar v_r)$ satisfy recursion \eqref{eq:recursion_tax} and be non-decreasing in the first argument, $r\in\mc{R}$. 
Consider the taxation mechanism $T$ that uses the perceived latency functions $\bar \ell_r(x; \bar v_r)$.	
	\begin{itemize}[leftmargin=5.5mm]
	\item[$\bullet$] Then, the price of anarchy of $T$ is no higher than $\rho=\sup_{\ell\in\mc{L}}\rho_\ell$.
	\item[$\bullet$] Further, for any choice of $\varepsilon>0$, one can compute \emph{in polynomial time} taxes whose price of anarchy is no-higher than $\rho+\varepsilon$. {The result holds for mixed Nash and correlated/coarse correlated equilibria.}
	\end{itemize}
\end{theorem}

\begin{proof}
We begin with proving the first bullet point in two steps. In the first step, we compare an arbitrary pure profile $a\in\mc{A}$ with a mixed profile $\bar y_{i}$ arising from the solution of the linear programming relaxation, and show that \( \sum_{i=1}^N\sum_{a'_i\in\mc{A}_i} \bar y_{i,a'_i}w_i[\bar{C}_i(a)-\bar{C}_i(a'_{i}, a_{-i})] \ge \SC(a) -\sum_{r\in\mc{R}} p_r(\bar v_r) \), where $\bar C_i(a) = \sum_{r\in a_i}\bar \ell_r(x_r(a),\bar v_r)$ represents a player's perceived cost when factoring in both the original latency and the taxes. In the second step, which is the challenging step, we prove that $\sum_{r\in\mc{R}} p_r(\bar v_r) \leq \rho\SC(\opt{a})$. This is sufficient to conclude, as \(
\sum_{i=1}^N\sum_{a'_i\in\mc{A}_i} \bar y_{i,a'_i}w_i[\bar{C}_i(a)-\bar{C}_i(a'_{i}, a_{-i})] \ge \SC(a) -\sum_{r\in\mc{R}} p_r(\bar v_r) \ge \SC(a)-\rho\SC(\opt{a}),
\)
combined with \emergencystretch 20em%
$0\ge \mb{E}_{a\sim\sigma}\left[\sum_{i=1}^N\sum_{a'_i\in\mc{A}_i} \bar y_{i,a'_i}w_i[\bar{C}_i(a)-\bar{C}_i(a'_{i}, a_{-i})]\right]$ for all mixed, correlated, or coarse correlated equilibria $\sigma$ (by linearity of expectation and equilibrium definition), implies $0\ge \SC(\NE{a})-\rho\SC(\opt{a})$.

To prove the first step, let $x_r=x_r(a)$ be the total weight imposed on resource $r$ by allocation $a$. As $\bar\ell_r$ is non-decreasing by assumption, $\sum_{i=1}^N\sum_{a'_i\in\mc{A}_i} \bar y_{i,a'_i}w_i[\bar{C}_i(a)-\bar{C}_i(a'_{i}, a_{-i})]$
$\ge$ $\sum_{i=1}^N w_i \sum_{r\in a_i} \bar\ell_r(x_r;\bar v_r)$ $- \sum_{i=1}^N\sum_{a'_i\in\mc{A}_i}\bar y_{i,a'_i}w_i \sum_{r\in a_{i}'} \bar\ell_r(x_r+w_i;\bar v_r)$. Re-arranging terms, this is equal to $\sum_{r\in\mc{R}} x_r\bar\ell_r(x_r;\bar v_r) - \sum_{r\in\mc{R}}\sum_{i=1}^N \bar v_{r,i} {w_i}\bar\ell_r(x_r+w_i;\bar v_r)$. Then, as $\bar \ell_r$ satisfies \eqref{eq:recursion_tax}, this equals $\sum_{r\in\mc{R}} c_r(x_r)-\sum_{r\in\mc{R}} p_r(\bar v_r)$, and as $\sum_{r\in\mc{R}} c_r(x_r)=\SC(a)$, $\sum_{i=1}^N\sum_{a'_i\in\mc{A}_i} \bar y_{i,a'_i}w_i[\bar{C}_i(a)-\bar{C}_i(a'_{i}, a_{-i})] \ge \SC(a) - \sum_{r\in\mc{R}} p_r(\bar v_r)$ as desired.

To prove $\sum_{r\in\mc{R}} p_r(\bar v_r) \leq \rho\SC(\opt{a})$, denote with \rm{LP}$(\bar z, \bar v, \bar y)$ the optimal value of \eqref{eq:cvxprogram}.
Let $P({\bar v_{r,i}})$ be independent Poisson random variables. Then, using the definition of $p_r(v_r)$ from \cref{eq:pdef}, Lemma~\ref{lem:main-technical} (whose statement is included after this proof), and noting that $c_r$ is a convex function, we get
\begin{align}
p_r(\bar v_r)
=
\mb{E} %
\left[c_r\left(\sum_{i\in [N]} w_i P({\bar v_{r,i}})\right)\right]
\le 
\mb{E} 
\left[c_r\left(P(1) \sum_{S\subseteq [N]} \left(\sum_{i\in S} w_i\right) \bar\chi_{S,r}\right)\right],
\label{eq:ExOverME}
\end{align}
where the $\bar\chi_{S,r}$ are mutual exclusive Bernoulli random variables with $Pr(\bar\chi_{S,r}=1) = {\bar z_{S,r}}$
and with probability 1, one and only one of the  $\bar\chi$'s is 1. 
In more detail, the latter inequality follows by recalling that if two distributions $P$ and $Q$ are such that $P\le_{cx} Q$, then for any convex function $f:\mb{R}\rightarrow\mb{R}$ it is $\mb{E}_P[f(P)] \le \mb{E}_ Q[f(Q)]$, see \cite{shaked2007stochastic}. Therefore, \eqref{eq:ExOverME} follows by taking the expected value of the convex function $c_r$, and applying the result in \cref{lem:main-technical}. 
It follows that
\begin{equation}\label{eq:ExOverME2}
\begin{aligned}
 {\mb{E} \left[c_r\left(P(1) \sum_{S\subseteq [N]} \left(\sum_{i\in S} w_i\right) \bar\chi_{S,r}\right)\right]}&=
 \sum_{S\subseteq [N]}  \mb{E} \left[\bar\chi_{S,r} \cdot c_r\left(P(1)  \left(\sum_{i\in S} w_i\right) \right)\right] \\
=  \sum_{S\subseteq [N]}  \mb{E} \left[\bar\chi_{S,r}\right] \cdot  \mb{E} \left[c_r\left(P(1)  \left(\sum_{i\in S} w_i\right) \right)\right] &= 
 \sum_{S\subseteq [N]}  \bar{z}_{S,r} \cdot  \mb{E} \left[c_r\left(P(1)  \left(\sum_{i\in S} w_i\right) \right)\right] \\
 &\le \rho  \sum_{S\subseteq [N]}  \bar{z}_{S,r} \cdot c_r\left(\sum_{i\in S} w_i\right) .
\end{aligned}
\end{equation}
Combining \eqref{eq:ExOverME} and \eqref{eq:ExOverME2} and summing over all $r\in \mc{R}$, we get
\begin{align}\label{eq:ExOverME3}
 \sum_{r\in \mc{R}} p_r(\bar v_r) &\le \rho  \sum_{r\in \mc{R}} \sum_{S\subseteq [N]}  \bar{z}_{S,r} \cdot c_r\left(\sum_{i\in S} w_i\right) 
=  \rho \text{LP}(\bar z, \bar v, \bar y)
\le \rho\SC(\opt{a}),
\end{align}
where the last step follows as $(\bar z, \bar v, \bar y)$ is optimal for the relaxation of \eqref{eq:cvxprogram}. 
This concludes the proof of the first bullet point of the theorem.

The proof of the second follows readily upon replacing an exact solution $(\bar z, \bar v, \bar y)$ of \eqref{eq:cvxprogram} with an approximate solution. In particular, using a 1+$\varepsilon/\rho$ approximate solution (which can be computed in polynomial time) and following the same steps of the previous proof will yield the desired price of anarchy of $\rho+\varepsilon$. 
\end{proof}

The proof of the previous theorem hinges upon a key inequality, which is shown in the following Lemma whose proof is deferred to \cref{sec:app-main-technical}.

\begin{lemma}
\label{lem:main-technical}
Fix $r\in \mc{R}$. Let $(z, v, y)$ be a feasible solution of the linear programming relaxation \eqref{eq:cvxprogram}. 
Let $(\chi_{S,r})_{S\subseteq[N]}$ be a vector of mutual exclusive %
Bernoulli random variables, s.t. $Pr(\chi_{S,r}=1) = z_{S,r}$. 
Let $(P(v_{r,i}))_{i\in[N]}$ and $(P(z_{S,r}))_{S\subseteq[N]}$ be two sets of independent Poisson random variables.  %
Let $P(1)$ be independent of the $\chi$'s.
Then 
\vspace{-8pt}
\[
   \sum_{i\in[N]} w_i\cdot P(v_{r,i})
    \; \le_{cx}\; \sum_{S\subseteq[N]} \left(\sum_{i\in S} w_i\right) P(z_{S,r}) 
   \; \le_{cx}\; P(1) \sum_{S\subseteq[N]}\left(\sum_{i\in S} w_i\right)  \chi_{S,r}.
\]
\end{lemma}

\subsection{Tight Polynomial Algorithms}
Since the result in \cref{thm:opttaxes} holds also for correlated/coarse correlated equilibria, we can leverage existing polynomial time algorithms to compute such equilibria and achieve an approximation matching the corresponding price of anarchy. This is a direct consequence of \cref{thm:opttaxes} and polynomial computability of correlated equilibria \cite{JiangL15}.

\begin{corollary}
\label{cor:polyalgo} 
Consider a weighted congestion game $G$ where each latency belongs to a common set of functions $\mc{L}$  and satisfies \cref{assumptionlatency} and \cref{assumptionbounded}. Let $(\bar z, \bar v, \bar y)$ be an optimal solution of \eqref{eq:cvxprogram}, let $\bar \ell_r(x; \bar v_r)$ satisfy recursion \eqref{eq:recursion_tax} and be non-decreasing in the first argument, for $r\in\mc{R}$. Let $\rho=\sup_{\ell\in\mc{L}}\rho_\ell$. For any $\varepsilon>0$, there exists a polynomial time algorithm to compute an allocation with a cost lower than \mbox{$(\rho +\varepsilon)\cdot \min_{a\in\mc{A}}\SC(a)$.} 
\end{corollary}

The approximation ratio presented in \cref{cor:polyalgo} can be achieved, for example, as follows. Given a desired tolerance $\varepsilon>0$, we design a taxation mechanism ensuring a price of anarchy of $\sup_{\ell\in\mc{L}} \rho_{\ell} +\varepsilon/2$, which can be done in polynomial time thanks to \cref{thm:opttaxes}. We use such taxation mechanisms, and compute an exact correlated equilibrium in polynomial time leveraging the result of Jiang and Leyton-Brown \cite{JiangL15}, who guarantees that the resulting correlated equilibrium has polynomial-size support. We then compute a correlated equilibrium and enumerate all pure strategy profiles in its support, identifying with $a^*$ that with lowest cost. Since the price of anarchy bounds of \cref{thm:opttaxes} hold for correlated equilibria, the pure strategy profile $a^*$ inherits a matching or better approximation ratio.

\section{Polynomial Latency Functions}
\label{sec:poly}
While \cref{thm:opttaxes} provides a methodology to derive taxation mechanisms with a price of anarchy matching the lower bound, it does not give an explicit expression for the taxation mechanism. In this section, we provide an explicit taxation mechanism (\cref{thm:recursionpolynomials}) for the widely studied class of polynomial latency functions $\ell_r(x)=\sum_{\d=0}^\D \ca_\d^rx^\d$ that provably satisfies recursion (\ref{eq:recursion_tax}) and has non-decreasing modified latency functions. We observe that for class $\mc{L}$ of non-negative linear combinations of $\{1,x,\dots,x^D\}$, the highest degree monomial $x^D$ determines $\rho = \max_{\ell \in \mc{L}} \rho_\ell$, which reduces to $B_{D+1}$, the $(D+1)$'st Bell number, as \mbox{we show in \cref{lem:rhopolynomial} in \cref{app:poly}.}

\begin{theorem}
\label{thm:recursionpolynomials}
Given a weighted congestion game $G$ with latency functions $\ell_r(x)=\sum_{\d=0}^\D \ca_\d^rx^\d$, $r\in\mc{R}$, consider the taxation mechanism $T$ which replaces the perceived latency functions of the agents with 
$\bar\ell_r(x;\bar v_r)=\ell_r(x)+\tau_r(x;\bar v_r)=\sum_{\d=0}^\D \ca_\d^rT_\d(x;\bar v_r)$, where  $(\bar z, \bar v, \bar y)$ solves the linear program in \eqref{eq:cvxprogram} and $T_{\d}(x;v)$ is defined as 
\[
T_\d(x;v)=\sum_{\k=0}^\d \mya{\k}{\d}(v)\cdot x^\k \qquad \text{with}\qquad
\begin{cases}
\mya{\d}{\d}(v)&=1\\
\mya{\j-1}{\d}(v)&=	\sum_{\k=\j}^d 
\binom{\k}{\j}
\mya{\k}{\d}(v)
\myb_{\k-\j}(v)\qquad \j\in[\d]
\end{cases}
\] 
and $\myb_{j}(v)=\sum_i v_i w_i^{j+1}$. 

\begin{itemize}[leftmargin=5.5mm]
	\item[$\bullet$] Then, the price of anarchy of $T$ is upper bounded by $B_{D+1}$.
	\item[$\bullet$] Further, for any choice of $\varepsilon>0$, one can compute \emph{in polynomial time} taxes whose price of anarchy is no-higher than $B_{D+1}+\varepsilon$.
	\end{itemize}
\end{theorem}

We note that for linear latencies $\ell(x)=x$, we obtain $T_1(x;v) = \sum_{i}v_iw_i + x$, and for quadratic latencies $\ell(x)=x^2$, we obtain $T_2(x;v) = \left(\sum_i v_iw_i\right)^2+2\left(\sum_i v_i w_i^2\right)+\left(\sum_i v_iw_i\right)x+x^2$. 

To illustrate how significantly the taxation mechanism improves the Price of Anarchy compared to the setting without interventions, we provide a comparison for polynomial latencies in Table \ref{table:poacomparison} in \cref{sec:poa-comparison}. 
\\

{{\bf Proof idea.} To prove the theorem, we will show that the modified latency functions $(\bar \ell_r)_{r\in\mc{R}}$ satisfy the conditions of \cref{thm:opttaxes}, i.e., satisfy the recursion in \eqref{eq:recursion_tax} and are non-decreasing. This will be sufficient to apply the result in \cref{thm:opttaxes} and thus obtain the two desired statements. Non-decreasingness follows readily by their expression. Regarding the satisfaction of \eqref{eq:recursion_tax}, there are three main ideas we employ. First, we observe that \eqref{eq:recursion_tax} can be proved ``monomial by monomial'', if we assume that the modified latencies $\bar \ell_r$ are obtained by linear combinations of the {\it same} coefficients defining $\ell_r$, that is, $\bar\ell_r(x;v)=\sum_{\d=0}^\D \ca_\d^rT_\d(x;v)$. Second, we note that for each fixed monomial it suffices to consider functions $T_d(x;v)$ that are polynomial of maximum degree $d$, that is $T_\d(x;v)=\sum_{\k=0}^\d \mya{\k}{\d}(v)\cdot x^\k$. Finally, we utilize this ansatz to derive a recursion for the coefficients $(\mya{\j}{\d}(v))_{\j=0}^\d$ that ensures \eqref{eq:recursion_tax} is indeed satisfied. 
This non-trivial and most technical part of the proof utilizes two independent lemmas.}

\begin{proof}
Regarding non-decreasingness: each term appearing in the recursion defining the coefficients $(\mya{\j}{\d}(\bar v_r))_{\j=0}^\d$  is non-negative since $\bar v_r$ is so. Thus all modified latencies $\bar\ell _r(x;\bar v_r)$ must be non-decreasing.
\medskip

Regarding the modified latency functions satisfying \eqref{eq:recursion_tax}: in the following we will prove \eqref{eq:recursion_tax} ``monomial by monomial''. That is, we will show that for all $x\ge0$, for all non-negative parameters vectors $v$ and weights $(w_1,\dots, w_N)$, and for all $\d\in[\D]$, it is 
\be
\label{eq:monomial-by-monomial}
x T_{\d}(x;v)-\sum_i  v_i{w_i}T_{\d}(x+w_i;v)=x^{\d+1}-M_{\d+1}(v),
\ee
where, to ease the notation, we have defined $M_{\d+1}(v)=\mb{E}\left[(\sum_i P_iw_i)^{\d+1} \right]$ with $P_i\sim\text{Poi}\left({v_{i}}\right)$ independent.\footnote{Note that $M_{\d+1}(v)$ is obtained by evaluating \eqref{eq:pdef} for a pure monomial latency $x^\d$. %
}
Evaluating the previous equality for $x=x_r$, $v_i=\bar v_{r,i}$, multiplying it by $\ca_{\d}^r$, and summing over $\d\in\{0,\dots,\D\}$, $r\in \mc{R}$, will yield \eqref{eq:recursion_tax}, thanks to linearity of expectation.

Hence, we turn our attention to proving \eqref{eq:monomial-by-monomial}. Towards this goal, we substitute the expression for $T_d(x;v)$, and group the terms depending on the degree with which $x$ appears. To do so, we utilize the binomial expansion for $(x+w_i)^\k=\sum_{\j=0}^\k \binom{\k}{\j} x^\j w_i^{\k-\j}$. To ease the notation, we drop the functional dependence of $\mya{\j}{\d}(v)$ and $M_{\d+1}(v)$ on $v$. We are left with
\[
x^{\d+1}(\mya{d}{d}-1) + \sum_{\j=1}^\d \mya{\j-1}{\d}x^\j
- \sum_{\k=0}^\d \sum_{\j=0}^\k \mya{\k}{\d}\binom{k}{j}x^\j \sum_i v_i w_i^{\k-\j{+1}}+M_{\d+1}=0,%
\]
which we rewrite exchanging the order of summation in the third term and using the definition of $\myb_j$
\[
(\mya{d}{d}-1)x^{\d+1} + \sum_{\j=1}^\d
\left(\mya{\j-1}{\d}
- \sum_{\k=\j}^\d \mya{\k}{\d}\binom{k}{j} \myb_{\k-\j}\right)x^\j - 
\sum_{\k=0}^\d \mya{\k}{\d} \myb_{\k}
+M_{\d+1}=0.
\]
By equating each coefficient multiplying the monomial $x^\j$, $\j \in[d+1]$, to zero, we obtain the recursion defining $(\mya{\j}{\d})_{\j=0}^\d$ given in the theorem statement. However this does not complete the proof as we still need to verify that the terms of order zero cancel out. Therefore, we are left to show that 
\be 
\sum_{\k=0}^\d \mya{\k}{\d} \myb_{\k}=M_{\d+1}.
\label{eq:degree0equal}
\ee
Proving this claim is non-trivial and is shown in two separate lemmas which might be of independent interest. In particular, \cref{lem:Md} (in \cref{app:poly}) will show that $M_{\d+1}$ defined above can be computed through a recursion akin to the binomial-based recursion defining the Bell numbers \cite[Eq 1.6.13]{wilf2005generatingfunctionology}. The crux of the argument is then developed in the following \cref{lem:degree0equal}, where this recursive definition is exploited to show that \eqref{eq:degree0equal} holds. This lemma generalizes a result from \cite{gould2007linear} where authors show that Bell numbers can be constructed by summing a sequence of numbers (different from the commonly-employed Stirling numbers of the second kind) naturally arising in combinatorics \cite[Eq 4.2]{gould2007linear}. 
\end{proof}

\begin{lemma}
\label{lem:degree0equal}
	Fix any $\d\in \mb{N}_0$, $v\in\mb{R}^n$ and let $M_{\d}$ and $(\mya{\d}{\j})_{\j=0}^\d$ be defined as in \cref{thm:recursionpolynomials}. Then \eqref{eq:degree0equal} holds.
\end{lemma}

\begin{proof}
We prove the statement in two parts. First, we show that $(\mya{\d}{\j})_{\j=0}^\d$ defined in \cref{thm:recursionpolynomials} can be equivalently defined through a different recursion. We then exploit this equivalent definition of $(\mya{\d}{\j})_{\j=0}^\d$ and \cref{lem:Md} (in \cref{app:poly}) to conclude.
\medskip

\noindent $\triangleright$ {\it First part:}~Given $\d\in \mb{N}_0$, $v\in\mb{R}^n$, we define $(\tildemya{\d}{\j})_{\j=0}^\d$ as
\[
\begin{cases}
\tildemya{\d}{\d}(v)&=1\\
\tildemya{\j}{\d}(v)&=	\sum_{\p=\j}^{d-1} 
\binom{\d}{\p+1}
\tildemya{\j}{\p}(v)
\myb_{\d-1-\p}(v)\qquad \j\in[0,\d],
\end{cases}
\] 
and show that the sequences $(\mya{\d}{\j})_{\j=0}^\d$ and $(\tildemya{\d}{\j})_{\j=0}^\d$ are identical. We do so by induction. Towards this goal, let $\mydelta = \{(j,d)
\in\mb{N}_0\times\mb{N}_0~\text{s.t.}~\j\le \d\}$. We first verify that equality holds for all $(j,d)\in\mydelta$ such that $j=d$. This is immediate to see since $\tildemya{\d}{\d}=1=\mya{\d}{\d}$. We then assume that the property $\tildemya{\j}{\d}=\mya{\j}{\d}$ holds for all $(\j,\d)\in \mydelta$ with $ \d-\ell \le j\le d$ and show that the property holds for all $(\j,\d)\in\mydelta$ with $ \d-(\ell+1) \le j\le d$. Once this is done, induction over $\ell\ge 0$ allows us to conclude. To complete the inductive step, we now consider $\mya{\j}{\d}$ with $\j=\d-(\ell+1)$, and utilize both its recursive definition and the inductive assumption to write 
\allowdisplaybreaks
\begin{align*}
\mya{\j}{\d}&=
\sum_{\k=\j+1}^\d \binom{\k}{\j+1} \mya{\k}{\d} \myb_{\k-(\j+1)}\\
(\text{$\mya{\k}{\d}=\tildemya{\k}{\d}$ by ass.%
})&=
\sum_{\k=\j+1}^{\d-1} \binom{\k}{\j+1} \tildemya{\k}{\d} \myb_{\k-(\j+1)}
{+\binom{\d}{\j+1} \tildemya{\d}{\d} \myb_{\d-(\j+1)}}
\\
(\text{$\tildemya{\d}{\d}=1$; def. of $\tildemya{\k}{\d}$)}&=
\sum_{\k=\j+1}^{\d-1} \sum_{\p=\k}^{\d-1}\binom{\k}{\j+1}
\binom{\d}{\p+1} \myb_{\k-(\j+1)}\myb_{(\d-1)-\p}
\tildemya{\k}{\p}
{+\binom{\d}{\j+1}  \myb_{\d-(\j+1)}}
\\
(\text{$\tildemya{\k}{\d}=\mya{\k}{\d}$ by ass. %
})
&=
\sum_{\k=\j+1}^{\d-1} \sum_{\p=\k}^{\d-1}\binom{\k}{\j+1}
\binom{\d}{\p+1} \myb_{\k-(\j+1)}\myb_{(\d-1)-\p}
\mya{\k}{\p}
{+\binom{\d}{\j+1} \myb_{\d-(\j+1)}}
\\
&=\sum_{(\k,\p)\in A_{\j,\d}}  \binom{\k}{\j+1}
\binom{\d}{\p+1} \myb_{\k-(\j+1)}\myb_{(\d-1)-\p}
\mya{\k}{\p}
{+\binom{\d}{\j+1} \myb_{\d-(\j+1)}},
%
\end{align*}
where we defined $A_{\j,\d}=\{(\k,\p)\in \Delta~\text{s.t.}~\k\ge \j+1, \p\le \d-1\}$. We proceed similarly for $\tildemya{\j}{\d}$, and write 
\begin{equation*}
\begin{split}
\tildemya{\j}{\d}&=
\sum_{\p=\j}^{\d-1} \binom{\d}{\p+1} \tildemya{\j}{\p} \myb_{(\d-1)-\p}\\
(\text{$\tildemya{\j}{\p}=\mya{\j}{\p}$ by ass.%
})&=
\sum_{\p=\j+1}^{\d-1} \binom{\d}{\p+1} \mya{\j}{\p} \myb_{(\d-1)-\p}
{+\binom{\d}{\j+1} \mya{\j}{\j} \myb_{\d-(\j+1)}}
\\
(\text{$\mya{\j}{\j}=1$; def. of $\mya{\j}{\p}$)}&=
\sum_{\p=\j}^{\d-1} \sum_{\k=\j+1}^{\p}\binom{\d}{\p+1}
\binom{\k}{\j+1}\myb_{(\d-1)-\p}  \myb_{\k-(\j+1)} 
\mya{\k}{\p}
{+\binom{\d}{\j+1}  \myb_{\d-(\j+1)}}
\\
&=\sum_{(\k,\p)\in A}  \binom{\k}{\j+1}
\binom{\d}{\p+1} \myb_{\k-(\j+1)}\myb_{(\d-1)-\p}
\mya{\k}{\p}
{+\binom{\d}{\j+1} \myb_{\d-(\j+1)}}.
\end{split}
\end{equation*}
This completes the proof as the two expressions match, as desired.
\medskip

\noindent $\triangleright$~{\it Second part}: Owing to the previous result, it is then sufficient to show that 
\[\sum_{\k=0}^\d \tildemya{\k}{\d} \myb_{\k}=M_{\d+1}.\]
Towards this goal, we rewrite the recursion for $\tildemya{\j}{\d}$ summing over $i=\d-\p-1$, instead of $p$. Hence, we have 
$\tildemya{\j}{\d}=\sum_{\i=0}^{\d-\j-1}\binom{\d}{\i}\tildemya{\j}{\d-\i-1} \myb_{\i}$. In the remainder of the proof, we use the recursion defining $M_d$ from \cref{lem:Md}.

Since we want to prove the claim by induction over $\d$, we first verify the claim for $d_0=0$. Using the recursion for $M_{\d+1}$ we obtain $M_{\d_0+1}=M_{1}=\myb_{0}M_0=\myb_0$, which equals the desired left-hand side $\sum_{\k=0}^{\d_0}\tildemya{\k}{\d_0} \myb_{\k}=\tildemya{0}{0} \myb_{0}=\myb_0$.

To carry out the inductive step, we now assume that 
\[\sum_{\j=0}^{\k-1} \tildemya{\j}{\k-1} \myb_{\j}=M_{\k}\quad \text{for all}\quad \k\le \d,\]
and want to show that the previous equation holds also for $\k=\d+1$. Towards this goal, we start from $M_{d+1}$ and use both its recursive definition and the inductive assumption 
\allowdisplaybreaks
\begin{align*}
M_{\d+1}&=
\myb_\d M_0+\sum_{\k=1}^{\d}\binom{\d}{\k} \myb_{\n-\k} M_\k\\
(\text{inductive assumption})&=
\myb_\d M_0+\sum_{\k=1}^{\d}\sum_{\j=0}^{\k-1}\binom{\d}{\k}\tildemya{\j}{\k-1}\myb_{\j}\myb_{\n-\k}\\
(\text{swap sums and let $\i=n-\k$})&=
\myb_\d M_0+\sum_{\j=0}^{\d-1}\sum_{\i=0}^{\d-\j-1} \binom{\d}{\d-\i}\tildemya{\j}{\d-1-\i}\myb_\i \myb_\j\\
(\text{$M_0=1=\tildemya{\d}{\d}$, and recursive definition of $\tildemya{\j}{\d}$})&=
\myb_\d\tildemya{\d}{\d}+\sum_{\j=0}^{\d-1}\tildemya{\j}{\d}\myb_j\\
&= \sum_{\j=0}^\d \tildemya{\j}{\d}\myb_j. \qedhere
\end{align*}
\end{proof}

%% file: parts-body/hardnessapprox.tex
\section{\NP-hardness of Approximation and Limitations of Unfair Interventions}
\label{sec:hardness}
Having shown a fair, efficiently-computable taxation mechanism with best-possible price of anarchy amongst all fair interventions, we investigate whether unfair interventions can achieve strictly better price of anarchy than their fair counterpart. Perhaps surprisingly, the answer is in the negative when polynomial computability is required for the intervention and the equilibrium concept. To obtain this result, we first show that minimizing the social cost is \NP-hard to approximate below a factor identical to the lower bound on the price of anarchy of fair interventions (\cref{thm:hardness}). Combined with the polynomial computability of correlated equilibria \cite{JiangL15}, and using \cite{Rough_barrier}, this immediately gives the following result:

\begin{corollary}
In weighted congestion game $G$ where each latency belongs to a common set of functions $\mc{L}$  and satisfies \cref{assumptionlatency}%
, no polynomially computable intervention, fair or unfair, achieves a price of anarchy less than $\rho_\ell = \sup_{\ell \in \mc{L}}$ for correlated equilibria, assuming \mbox{\Pclass $\neq$ \NP}.
\end{corollary}
Note that the polynomial time requirement is necessary as otherwise the intervention could just compute the optimum and prescribe the users to play accordingly.   

We further point out that minimizing the social cost in weighted congestion games belongs to the more general class of \emph{Optimization Problems with Diseconomies of Scale} studied by Makarychev and Sviridenko \cite{Makarychev18}. Hence, the ensuing theorem not only establishes latency function dependent hardness results for weighted congestion games, but also is the first to show tightness for the broader class of optimization problems with diseconomies of scale. Specifically, our result shows that the approximation algorithm in \cite{Makarychev18} is the best possible already for very simple settings, i.e., optimization problems with diseconomies of scale where the constraint set is a simplex, all resources have the same latency function, and all weights are equal. This was not known, and we believe may be of independent interest. 

\begin{theorem}
\label{thm:hardness}
In weighted congestion games where all resources feature the same latency function $\ell:\mb{R}_{\mygezero}\rightarrow\mb{R}_{\mygezero}$ and $\ell$ is non-decreasing semi-convex, 
\mincon~is \NP-hard to approximate within any factor smaller than $\rho_\ell$ as defined in (\ref{eq:approxfact}), i.e.
\[ \rho_\ell=\sup_{t>0}\frac{\mb{E}_{P\sim\text{Poi}(1)} [tP\ell(tP)]}{t\ell(t)}. \]

{If $\rho_\ell=\infty$, then \mincon~is \NP-hard to approximate within any finite factor.
  These NP-hardness results hold even if all players have the same weight $w$.
 }
\end{theorem}

Naturally, \cref{thm:hardness} applies directly to richer classes of congestion games, whereby the players' weights or the resources' latency functions can differ. For example, if the latency function of each resource can be constructed by non-negative linear combination of given functions $\{\ell_1,\dots,\ell_m\}$, then \mincon~is \NP-hard to approximate within any factor smaller than $\max_{j\in \{1,\dots,m\}}\rho_{\ell_j}$, i.e., smaller than that produced by the worst function. 

\begin{proof}
We will prove the theorem in two steps:
\begin{enumerate}
\item[(a)] 
First we will show that for games where all players have the same arbitrary but fixed weight $w$ and all all resources have the same latency function $\ell$, \mincon\ is \NP-hard to approximate within any factor smaller than 
\begin{align*}
\sup_{x\in\mb{N}}\frac{\mb{E}_{P\sim\text{Poi}(x)}[wP\ell(wP)]}{wx\,\ell(wx)}.
\end{align*}
\item[(b)]
Afterwards, we show that for the worst case $w$,  i.e.:
\begin{align}
w & \in \arg \sup_{t>0}\frac{\mb{E}_{P\sim\text{Poi}(1)} [tP\ell(tP)]}{t\ell(t)},
\label{eq:wmax}
\end{align}
this factor reduces to the one given in the theorem.
\end{enumerate}

We start by showing (a): 
Recall that we have fixed $w$. Since all players have the same weight, the latency function $\ell$ will only have to be evaluated at integer multiples of $w$. For $x\in \mathbb{N}$, define a new latency function $\ell'$ as $\ell'(x) = \ell(x\cdot w)$. To complete the proof of (a), we reduce from \emph{unweighted} congestion games where all resources have the same latency function $\ell'$ to weighted congestion games where all players have weight $w$ and all resources have the same latency function $\ell$. For \emph{unweighted} congestion games where all resources have the same latency function $\ell'$ it has been shown \cite{paccagnan2021optimal} that it is NP-hard to approximate \mincon\ by any factor smaller than 
\begin{align*}
\rho_{\ell'}=\sup_{x\in\mb{N}}\frac{\mb{E}_{P\sim\text{Poi}(x)}[P\ell'(P)]}{x\,\ell'(x)}.
\end{align*}
Given an \emph{unweighted} congestion game $G_u$ where all resources have the same latency function $\ell'$, construct the corresponding \emph{weighted} congestion game $G_w$ by changing all player weights from $1$ to $w$ and changing all latency functions from $\ell'$ to $\ell$. The set of players and resources and the strategy space of the players remains unchanged. This is clearly a polynomial construction. 

For any allocation $a\in \mc{A}$ denote $\SC(G_w,a)$ and $\SC(G_u,a)$ the corresponding social cost of the allocation $a$ in $G_w$ and $G_u$, respectively. Let $x_r=|a|_r$ be the number of players assigned by allocation $a$ to resource $r$. Observe that by construction we have 
\begin{align}
   \SC(G_w,a) 
   = \sum_{r\in R} w x_r \,\ell(w x_r) 
   = \sum_{r\in R} w x_r \,\ell'(x_r) 
   = {w} \sum_{r\in R} x_r \,\ell'(x_r)
   = {w} \cdot \SC(G_u,a),
   \label{eq:sc-unweigted-vs-weighted}
 \end{align}
that is, for all allocations $a\in\mc{A}$ the corresponding social costs in $G_w$ and $G_u$ differ by the fixed scalar factor $w$.
This implies that an optimal allocation $a^*$ for $G_w$ is also optimal for $G_u$ and for any allocation $a\in \mc{A}$ we have
\begin{align}
\frac{\SC(G_w,a) }{\SC(G_w,a^*)}
=\frac{\SC(G_u,a) }{\SC(G_u,a^*)}.
\label{eq:SCratios}
\end{align}
It follows that for weighted congestion games where all players have the same weight $w$ and all resources have the same latency function $\ell$, \mincon\ is also NP-hard to approximate by any factor smaller than 
\begin{align*}
\rho_{\ell'}
=\sup_{x\in\mb{N}}\frac{\mb{E}_{P\sim\text{Poi}(x)}[P\ell'(P)]}{x\,\ell'(x)}
=\sup_{x\in\mb{N}}\frac{\mb{E}_{P\sim\text{Poi}(x)}[wP\ell(wP)]}{wx\,\ell(wx)}.
\end{align*}
\smallskip
This finishes the proof of statement (a).

To complete the proof, we will show (b). That is for $w$ defined as in \eqref{eq:wmax}, we have 
\[
\sup_{x\in\mb{N}}\frac{\mb{E}_{P\sim\text{Poi}(x)}[wP\ell(wP)]}{wx\,\ell(wx)}
= \sup_{t>0}\frac{\mb{E}_{P\sim\text{Poi}(1)} [tP\ell(tP)]}{t\ell(t)}.
\]
We do so by proving that the left hand side is simultaneously larger-equal and smaller-equal to the right hand side. Towards proving the first of these two statements observe that 
\[
\sup_{x\in\mb{N}}\frac{\mb{E}_{P\sim\text{Poi}(x)}[wP\ell(wP)]}{wx\,\ell(wx)}\ge
\frac{\mb{E}_{P\sim\text{Poi}(1)}[wP\ell(wP)]}{w\,\ell(w)}
=
\sup_{t>0}\frac{\mb{E}_{P\sim\text{Poi}(1)} [tP\ell(tP)]}{t\ell(t)},
\]
where we used the definition of $w$ from \eqref{eq:wmax}.

We now turn attention to the converse inequality. To do so, we will show the  slightly stronger inequality where we restrict the supremum on the right-hand-side of to integer multiples of $w$. That is, we show that
\begin{align}
\sup_{x\in\mb{N}}\frac{\mb{E}_{P\sim\text{Poi}(x)}[wP\ell(wP)]}{wx\,\ell(wx)}
\le \sup_{x\in\mb{N}}\frac{\mb{E}_{P\sim\text{Poi}(1)} [wxP\ell(wxP)]}{wx\ell(wx)}.
\label{eq:stochasticorder-hardness-strong}
\end{align}
We proceed by applying Theorem 3.A.36 in \cite{shaked2007stochastic} (which we report as \cref{claim:storder} in \cref{app:claimstochorder} for completeness) with $k=x$, 
$X_i \sim x\,\text{Poi}(1)$, $Y \sim x\,\text{Poi}(1)$, and $\alpha_i = \frac{1}{x}$.
Clearly, $X_i =_{st} Y$, and hence $X_i \leq_{cx} Y$, for $i=1,\dots,x$.  
Moreover, $\alpha_i \geq 0$ and $\sum_{i=1}^x \alpha_i = 1$.

It then follows that $\sum_{i=1}^x \alpha_i X_i \leq_{cx} Y$. 
As the sum of $x$ independent Poisson random variables with distribution $\text{Poi}(1)$ is a Poisson random variable with distribution $\text{Poi}(x)$, $\sum_{i=1}^x \frac{1}{x} X_i =  \sum_{i=1}^x \text{Poi}(1) =  \text{Poi}(x)$ and $\sum_{i=1}^x \alpha_i X_i \leq_{cx} Y$ implies $ \text{Poi}(x) \leq_{cx} x \text{Poi}(1)$. Since $w>0$ we get for any $x\in\mb{N}$,
$$w \text{Poi}(x) \leq_{cx} wx \text{Poi}(1).$$

From the definition of convex order and since $\ell$ is semi-convex, we get for any $x\in\mb{N}$,
\begin{align*}
\mathbb{E}_{P \sim \text{Poi}(x)} \left[w P \ell\left(w P\right)\right] 
\leq \mathbb{E}_{P \sim \text{Poi}(1)}[wxP\ell(wxP)],
\end{align*}
which immediately implies \eqref{eq:stochasticorder-hardness-strong} because of the common denominator.
This completes the proof of (b) and the theorem.
\end{proof}

%% file: parts-appendix/appendix.tex
\appendix
\section{Appendix}

\subsection{A Useful Lemma on Stochastic Orders}
\label{app:claimstochorder}
\begin{claim}(Theorem 3.A.36 in \cite{shaked2007stochastic})
Let $X_1,\dots,X_k$ and $Y$ be $k+1$ random variables and $X_i \leq_{cx} Y$, then $\sum_{i=1}^k \alpha_i X_i \leq_{cx} Y$ for $\alpha_i \geq 0$ and $\sum_{i=1}^k \alpha_i = 1$.
\label{claim:storder}
\end{claim}

\subsection{Proof of \cref{lem:main-technical}}
\label{sec:app-main-technical}
\begin{proof}
Let $z_{S,r,i}=z_{S,r}$ for all $i,S,r$ and introduce a set of independent Poisson random variables $P(z_{S,r,i})$. 
For every set $S$, denote $W_S=\sum_{i\in S} w_i$.
We start proving the left  inequality.  

To do so let us first fix a set $S$ and apply Claim~\ref{claim:storder}, 
with $k=|S|$,  
$Y = W_S P(z_{S,r})$, and for all $i\in S$,
$X_i= W_S P( z_{S,r,i})$ and $\alpha_i=\frac{w_i}{W_S}$. Clearly, $\alpha_i\ge 0$ and $\sum_{i\in S} \alpha_i =1$. 
Moreover, $X_i =_{st} Y$ and thus  $X_i \le_{cx} Y$ for all $i\in S$.
By Claim~\ref{claim:storder}, we get 
\begin{align*}
\sum_{i\in S} w_i \cdot P( z_{S,r,i})  
\le_{cx} 
\left(\sum_{i\in S} w_i\right) \cdot P( z_{S,r}),
\end{align*}
which implies 
\begin{align}
 \sum_{S\subseteq [N]}  \sum_{i\in S} w_i \cdot P( z_{S,r,i}) 
\le_{cx} 
 \sum_{S\subseteq [N]}  \left(\sum_{i\in S} w_i\right) \cdot P( z_{S,r}).
\label{eq:conPz}
\end{align}
We are now ready to show the inequality. Indeed, 
\begin{align*}
 \lefteqn{\sum_{i\in[N]} w_i\cdot P(v_{r,i})}\\
 & =_{st}  \sum_{i\in[N]} w_i\cdot  \sum_{S\subseteq [N]:i\in S} P( z_{S,r,i}) && \quad \text{(independence of $P$'s and constraint in \eqref{eq:cvxprogram})}\\
 & =_{st}  \sum_{S\subseteq [N]} \left(\sum_{i\in S} w_i \cdot P( z_{S,r,i}) \right) &&  \quad \text{(changing order of summation).}\\
 & \le_{cx}  \sum_{S\subseteq [N]} \left(\sum_{i\in S} w_i\right) \cdot P( z_{S,r}) &&  \quad \text{(by \eqref{eq:conPz})}.
\end{align*}
This finishes the proof of the first inequality of the lemma.

The second inequality, $ \sum_{S\subseteq[N]} \left(\sum_{i\in S} w_i\right) P(z_{S,r}) 
    \le_{cx} P(1) \sum_{S\subseteq[N]}\left(\sum_{i\in S} w_i\right)  \chi_{S,r}$ has been shown in more general form by Makarychev, Sviridenko \cite[Inequality (35)]{Makarychev18}. 
This completes the proof of the lemma.
\end{proof}

\subsection{Additional material for Section~\ref{sec:taxes}}
\label{app:taxes}
\begin{lemma}\label{app:LPcomputation}
LP \eqref{eq:cvxprogram} can be solved efficiently to any precision in polynomial time.
\end{lemma}
\begin{proof}
Denote $\mc{P}$ the polytope defined by inequalities 
\be 
\begin{aligned}
&v_{r,i} = \sum_{a_i\in\mc{A}_i:r\in a_{i}} y_{i,a_i} && \text{for all } r\in \mc{R}, i\in[N],\\
&\sum_{a_i\in\mc{A}_i} y_{i,a_i} = 1 && \text{for all } i \in [N], \\
&y_{i,a_i} \ge 0 && \text{for all } i \in [N], a_i\in\mc{A}_i. \\
\end{aligned}
\ee

Let denote $H_r(y,v)$ for a given $(y,v)\in\mc{P}$ the minimal cost of the term of resource $r$ in the objective of LP \eqref{eq:cvxprogram}, i.e. the optimal cost of the following program whereas $(y,v)\in\mc{P}$ and $r\in\mc{R}$ are fixed
\vspace{-5pt}
\be 
\label{eq:polytopeprogram}
\begin{aligned}
\min
\quad & 
\sum_{S\subseteq [N]}
c_r\left(\sum_{i\in S}w_i \right) \cdot z_{S,r}\\
\text{\normalfont subject to}\quad  
& \sum_{S\subseteq [N]} z_{S,r} = 1 &&\\
& \sum_{S\subseteq [N]:i\in S} z_{S,r} = v_{r,i} &&\text{for all } i\in[N],\\
& z_{S,r}\ge 0 && \text{for all } S\subseteq[N]
\end{aligned}
\ee 
LP \eqref{eq:cvxprogram} can be equivalently written as
\vspace{-5pt}
\be
\begin{aligned}
\min
\quad & 
\sum_{r\in\mc{R}}
H_r(y,v) \\
\text{\normalfont subject to}\quad  
& (y,v) \in \mc{P} &&
\label{eq:conH}
\end{aligned}
\ee

Similar to \cite[Lemma 3.1]{Makarychev18}, we show that functions $H_r(y,v)$ are convex and that there exists a polynomial time algorithm to compute $H_r$ and find a subgradient of $H_r$. Together, this implies that the minimum of the convex program \eqref{eq:conH} can be found using the ellipsoid method. We recall the assumption that $w_i/\delta$ is a polynomially bounded integer. For simplicity, we assume $\delta=1$ and hence $w_i$ is integral (and polynomially bounded).

To show that $H_r$ is convex, with similar approach to \cite[Lemma E.1]{Makarychev18}, let $z^\star$ and $z^{\star\star}$ be optimal solutions for $(y^\star,v^\star)$ and $(y^{\star\star},v^{\star\star})$. As LP \eqref{eq:polytopeprogram} only has linear constraints, $\lambda z^\star+(1-\lambda)z^{\star\star}$ is feasible for $\lambda (y^\star,v^\star)+(1-\lambda)(y^{\star\star},v^{\star\star})$, for any $\lambda \in [0,1]$. Hence, $H_r(\lambda y^\star+(1-\lambda)y^{\star\star},\lambda v^\star+(1-\lambda)v^{\star\star})\le \lambda H_r(y^\star,v^\star)+(1-\lambda)H_r(y^{\star\star},v^{\star\star})$, as the LHS is at most the cost of $\lambda z^\star+(1-\lambda)z^{\star\star}$, which equals the expression on the RHS.

To show that $H_r$ can be computed efficiently, we obtain the dual LP to LP \eqref{eq:polytopeprogram} and introduce variable $\xi$ for constraint $\sum_{S\subseteq [N]} z_{S,r} = 1$ and variables $\eta_i$ for constraints $\sum_{S\subseteq [N]:i\in S} z_{S,r} = v_{r,i}$ for $i\in[N]$. 
\be
\begin{aligned}
\max
\quad & 
\xi +\sum_{i\in[N]} \eta_i v_{r,i} \\[-5pt]
\text{\normalfont subject to}\quad  
& \xi+\sum_{S\subseteq [N]:i\in S}\eta_i \le c_r\left(\sum_{i\in S}w_i\right) && \text{for all } S\subseteq[N] 
\end{aligned}
\ee 
This dual has exponentially many constraints. However, with argument as in \cite{Makarychev18}, we can guess $B^\star=\sum_{i\in S^\star}w_i$ for set $S^\star$ violating the constraint, as $w_i$ are polynomially bounded and hence $B^\star$. A maximum knapsack problem of size $B^\star$ over $S\subseteq [N]$ to maximise the sum of $\eta_i$ in the set is used to obtain the optimal set $S^\star$. Let $(\xi^\star,\eta^\star)$ be optimal for the dual LP corresponding to $v_{r,i}^\star$. We observe that due to strong duality, $H_r(y^\star,v^\star)$ equals the objective value of the dual LP. 

Further, for any feasible $v_{r,i}$, let $z_{S,r}$ the optimal variables for $H_r(y,z)$ and the constraints of the LPs imply
\be
\begin{aligned}
H_r(y,v) &= \sum_{S\subseteq [N]}z_{S,r}c_r\left(\sum_{i\in S}w_i\right)=\xi^\star+\sum_{S\subseteq [N]}z_{S,r}\left(c_r\left(\sum_{i\in S}w_i\right)-\xi^\star\right) \\ &\ge \xi^\star+\sum_{S\subseteq [N]}z_{S,r}\left(\sum_{i\in S}\eta^\star_i\right)=\xi^\star + \sum_{i\in[N]}\eta^\star_i \left(\sum_{S\subseteq [N]:i\in S}z_{S,r}\right)\\
&=\xi^\star+\sum_{i\in [N]}\eta^\star_i v_{r,i}
\end{aligned}
\ee
$H_r(y,v)\ge\xi^\star+\sum_{i\in[N]}\eta^\star_i v_{r,i}$ for any $v_{r,i}$ and therefore $\eta^\star$ is a subgradient of $H_r$ at $v^\star_{r,i}$.
\end{proof}

\subsection{Additional material for Section~\ref{sec:poly}} 
\label{app:poly}
\begin{lemma} \label{lem:rhopolynomial}
Let $D\in\mb{N}_0$ and let $\ell:\mb{R}_{\ge0}\rightarrow\mb{R}$ be $\ell(x)=\sum_{j=0}^D b_j x^j$ with $b_j\geq0$ for all $j$. Let $m:\mb{R}_{\ge0}\rightarrow\mb{R}$ be $m(x)=x^D$. Define $\rho_\ell,\rho_m$ as in (\ref{eq:approxfact}). Then $\rho_\ell \leq \rho_m$, and $\rho_m=B_{D+1}$, i.e., the $(D+1)$'st Bell number.
\end{lemma}
\begin{proof}
To show that \[\rho_\ell = \sup_{w>0}\frac{\mb{E}_{P\sim\text{Poi}(1)} [\sum_{j=0}^D b_j w^{j+1}P^{j+1}]}{\sum_{j=0}^D b_j w^{j+1}} \leq \sup_{w>0}\frac{\mb{E}_{P\sim\text{Poi}(1)} [w^{D+1}P^{D+1}]}{w^{D+1}} = \rho_m,\] we show that for any $w>0$, 

\[\sum_{j=0}^D b_j \mb{E}_{P\sim\text{Poi}(1)}[w^{j+1}P^{j+1}]w^{D+1}\leq \sum_{j=0}^D b_jw^{j+1}\mb{E}_{P\sim\text{Poi}(1)}[w^{D+1}P^{D+1}].\] 
The lemma then follows by rearranging and taking the supremum on both sides. Note that 
\[
\begin{aligned}
\sum_{j=0}^D b_j \mb{E}_{P\sim\text{Poi}(1)}[w^{j+1}P^{j+1}]w^{D+1} &\leq \sum_{j=0}^D b_jw^{j+1}\mb{E}_{P\sim\text{Poi}(1)}[w^{D+1}P^{D+1}] \iff \\
\sum_{j=0}^D b_j w^{j+D+2}\mb{E}_{P\sim\text{Poi}(1)}[P^{j+1}] &\leq \sum_{j=0}^D b_jw^{j+D+2}\mb{E}_{P\sim\text{Poi}(1)}[P^{D+1}].
\end{aligned}
\] Thus, to conclude it is sufficient to show for all $j \in \{0,\dots,D\}$, $\mb{E}_{P\sim\text{Poi}(1)}[P^{j+1}]\leq \mb{E}_{P\sim\text{Poi}(1)}[P^{D+1}]$. As the $(j+1)$'st moment of the Poisson distribution with parameter $1$ equals the $(j+1)$'st Bell number \cite{mansour2015commutation}, this statement is equivalent to showing that $B_{j+1}\leq B_{D+1}$ for $j\in\{0,\dots,D\}$, which holds as the Bell numbers are non-decreasing \cite{mansour2015commutation}.  
\end{proof}

\begin{lemma}
\label{lem:Md}
For every $\d\in\mb{N}_0$, $v\in\mb{R}^N$, and non-negative weights $(w_i)_{i\in[N]}$, define $M_{\d}=\mb{E}\left[(\sum_i P_iw_i)^{\d} \right]$, where $P_i\sim\text{Poi}\left({{v_{i}}}\right)$ are independent random variables. Let  $\myb_{\j}=\sum_i v_i w_i^{\j{+1}}$, $\j\in \mb{N}_0$. Then the sequence $(M_{d})_{d\in\mb{N}_0}$ can be recursively computed from
\[
\begin{cases}
M_0&=1\\
M_{\n+1}&=\sum_{\k=0}^{\n}  \binom{\n}{\k}\myb_{\n-\k}M_\k\qquad n\in\mb{N}	
\end{cases}
\]	
\end{lemma}
\begin{proof}
To prove the result we observe that $M_{\d}$ represents the $\d$-th moment of the random variable $P=\sum_i P_i w_i$ where $P_i\sim\text{Poi}\left(v_i\right)$ are independent. Therefore, we base our proof on the moment generating function approach, which allows us to easily compute the moments of $P$ by simply evaluating the derivatives of its moment generating function $\M_P(t)$ in the origin. 
Since $P_i\sim\text{Poi}(v_i)$, it is $\M_{P_i}(t)=e^{v_i(e^t-1)}$, see \cite{taboga2012probability}. Hence, thanks to the properties of the moment generating function \cite[Proposition 214]{taboga2012probability},  we have 
\[
\M_P(t)=\prod_i \M_{P_i}(w_it) = \prod_i e^{v_i(e^{w_it}-1)}%
=e^{\sum_i f_i(t)}\cdot e^{-\sum_i v_i},
\]
where we defined $f_i(t)={v_i e^{w_i t}}$. At this stage, we are left to compute the derivatives of $\M_P(t)$ since the $(n+1)$-st moment of $P$ is equal to the $(n+1)$-st derivative of $\M_P(t)$ evaluated in zero, i.e., $M_{n+1} = \M^{(n+1)}_P(0)$, see \cite[Prop. 210]{{taboga2012probability}}. We now claim that for any $n\in\mb{N}_0$ it is
\be
\M^{(n+1)}_P(t)=\sum_{\k=0}^n \binom{n}{k}\hat{\myb}_{n-\k}(t)\M^{(\k)}_P(t),
\label{eq:finalM}
\ee
where $\M_p^0(t)=\M_p(t)$ and $\hat \myb_j(t)=\sum_i f_i(t) w_i^{j+1}$, $j\in\mb{N}_0$. Once \eqref{eq:finalM} is shown, the claim follows readily by evaluating $\M^{(n+1)}_P(0)$, which will yield the desired recursion over $n$.  We thus turn the attention to \eqref{eq:finalM}, which we prove by induction. Towards this goal, we drop the functional dependence on $t$ everywhere to simplify the notation. It is immediate to verify that for $n=1$ the claim holds, since differentiating $\M_P$ gives
\[
\M'_P=%
\M_P\cdot  \sum_i w_i f_i = \binom{0}{0}\hat \myb_{0}\M^{(0)}_P.
\]
We now show that if the claim holds for a given $n\in\mb{N}_0$, it holds for $n+1$. Towards this goal, we start from the expression for $\M_P^{(n)}$, taken from \eqref{eq:finalM} with $n$ in place of $n+1$, which we differentiate to obtain $\M_P^{(n+1)}$. We then show that such expression equals \eqref{eq:finalM} as desired:
\begin{equation*}
\begin{split}
\M_P^{n+1}
&=\sum_{\k=0}^{n-1} \binom{n-1}{k}\left[\hat{\myb}'_{n-1-\k}\M^{(\k)}_P+ \hat{\myb}_{n-1-\k}\M^{(\k+1)}_P\right]\\
&=\sum_{\k=0}^{n-1} \binom{n-1}{k}\left[\hat{\myb}_{n-\k}\M^{(\k)}_P+ \hat{\myb}_{n-1-\k}\M^{(\k+1)}_P\right]\\
&=\binom{n-1}{0}\hat \myb_n \M_P^{(0)} + \sum_{\k=1}^{n-1}\left[\binom{n-1}{\k}+\binom{n-1}{\k-1}\right]\hat \myb_{n-\k}\M_P^{(\k)}+\binom{n-1}{n-1}\hat \myb_0 \M_P^{(n)}\\
&=\binom{n}{0}\hat \myb_n \M_P^{(0)} + \sum_{\k=1}^{n-1}\binom{n}{\k}\hat \myb_{n-\k}\M_P^{(\k)}+\binom{n}{n}\hat \myb_0 \M_P^{(n)}\\
&=\sum_{\k=0}^n \binom{n}{k}\hat{\myb}_{n-\k}\M^{(\k)}_P,
\end{split}
\end{equation*}
where the second line holds since $\hat{\myb}'_{n-1-\k}=\hat{\myb}_{n-\k}$, and the fourth line holds due to the fact that $\binom{n-1}{\k}+\binom{n-1}{\k-1}=\binom{n}{\k}$. This completes the proof.
\end{proof}

\begin{lemma}
\label{lem:factor-weighted-unweighted}
Let $\mc{L}$ be a class of semi-convex latency functions that is closed under abscissa scaling. Then approximating the minimum social cost is \NP-hard within the same factor in weighted and unweighted congestion games with latency functions in $\mc{L}$. 
\end{lemma}
\begin{proof}
Let $\ell \in \mc{L}$. Then, due to abscissa scaling, the function $\ell'(x) = \ell(\alpha x)$ for any $\alpha >0$ is also in $\mc{L}$. Let 
\[ \rho_u = \sup_{\ell \in \mc{L}} \sup_{x\in\mb{N}} \frac{\mb{E}_{P\sim\text{Poi}(x)} [P\ell(P)]}{x\ell(x)} \]
be the inapproximability factor in unweighted congestion games for class $\mc{L}$ \cite{paccagnan2021optimal}
and let \[ \rho = \sup_{\ell \in \mc{L}} \sup_{t >0} \frac{\mb{E}_{P\sim\text{Poi}(1)} [tP\ell(t P)]}{t\ell(t)} \]
be the inapproximability factor in weighted congestion games. Clearly, $\rho_u \leq \rho$ (unweighted congestion games are a subset of weighted congestion games) and it remains to show that $\rho_u \geq \rho$. This follows readily as
\begin{align*}
\rho_u = \sup_{\ell \in \mc{L}} \sup_{x\in\mb{N}} \frac{\mb{E}_{P\sim\text{Poi}(x)} [P\ell(P)]}{x\ell(x)} & \geq \sup_{\ell \in \mc{L}} \frac{\mb{E}_{P\sim\text{Poi}(1)} [P\ell(P)]}{\ell(1)} \\
& \geq \sup_{\ell \in \mc{L}, \alpha >0} \frac{\mb{E}_{P\sim\text{Poi}(1)} [P\ell(\alpha P)]}{\ell(\alpha)} \\
& = \sup_{\ell \in \mc{L}, \alpha >0} \frac{\mb{E}_{P\sim\text{Poi}(1)} [\alpha P\ell(\alpha P)]}{\alpha \ell(\alpha)} = \rho
\end{align*}
and the second inequality follows as $\mc{L}$ is closed under abscissa scaling.
\end{proof}

\subsection{Price of Anarchy with and without intervention for polynomial latencies}
\label{sec:poa-comparison}
\begin{table}[h]
\centering
\begin{threeparttable}
\begin{tabular}{ccc}
\toprule
Class of polynomial latency & \multicolumn{2}{c}{Price of Anarchy} \\
\cmidrule(r){2-3}
functions of degree & No Intervention \cite{aland2006exact} & Taxation mechanism [Thm. \ref{thm:recursionpolynomials}] \\
\cmidrule(lr){1-1}
\cmidrule(lr){2-2}
\cmidrule(lr){3-3}
1 & 2.618 & 2 \\
2 & 9.909 & 5 \\
3 & 47.82 & 15 \\
4 & 277 & 52 \\
5 & 1858 & 203 \\
6 & 14099 & 877 \\
7 & 118926 & 4140 \\
8 & 1 101 126 & 21 147 \\
9 & 11 079 429 & 115 975 \\
10 & 120 180 803 & 678 570 \\
\bottomrule
\end{tabular}
\begin{tablenotes}
\small
\item Price of Anarchy in the case of no interventions by \cite{aland2006exact}.
\end{tablenotes}
\end{threeparttable}
\caption{Comparison of the Price of Anarchy with and without interventions}
\label{table:poacomparison}
\end{table}